\documentclass[10pt]{article}
\usepackage{fullpage}
\usepackage[bookmarks, colorlinks=true, allcolors=blue]{hyperref}
\usepackage{amssymb}
\usepackage{amsmath}
\usepackage{amsthm}
\usepackage[abs]{overpic}
\usepackage{tcolorbox}
\usepackage{graphicx,color,colordvi}
\usepackage{bbm}
\usepackage{cleveref}
\usepackage{stmaryrd}
\usepackage[utf8]{inputenc}
\usepackage{fullpage}
\usepackage[blocks]{authblk}
\usepackage{dsfont}
\usepackage{mathtools}
\usepackage{authblk}

\usepackage{centernot}
\usepackage{pgfplots}

\def\openone{\leavevmode\hbox{\small1\kern-3.8pt\normalsize1}}

\def\RR{\mathbb{R}}

\def\11{\mathbb{I}}

\newcommand{\faf}[1]{%
  ^{\underline{#1}}%
}

\newtheorem{definition}{Definition}[section]
\newtheorem{proposition}[definition]{Proposition}
\newtheorem{lemma}[definition]{Lemma}

\newtheorem{theorem}[definition]{Theorem}
\newtheorem{corollary}[definition]{Corollary}
\newtheorem{conjecture}[definition]{Conjecture}

\newcommand{\tr}{\mathop{\rm Tr}\nolimits}

\usepackage{pst-node}
\usepackage{tikz-cd}

\newcommand{\cF}{{\mathcal{F}}}

\def\d{\mathrm{d}}

\newcommand{\reg}{{\mathrm{reg}}}

\usepackage{graphicx}
\usepackage{setspace}
\usepackage{verbatim}
\usepackage{subfig}

\numberwithin{equation}{section}

\DeclareRobustCommand\openone{\leavevmode\hbox{\small1\normalsize\kern-.33em1}}

\newcommand{\id}{{\rm{id}}}
\newcommand{\be}{\begin{equation}}
	\newcommand{\ee}{\end{equation}}
\newcommand{\bea}{\begin{eqnarray}}
	\newcommand{\eea}{\end{eqnarray}}
\newcommand{\beas}{\begin{eqnarray*}}
	\newcommand{\eeas}{\end{eqnarray*}}

\DeclareFontFamily{U}{mathx}{\hyphenchar\font45}
\DeclareFontShape{U}{mathx}{m}{n}{<-> mathx10}{}
\DeclareSymbolFont{mathx}{U}{mathx}{m}{n}
\DeclareMathAccent{\widebar}{0}{mathx}{"73}
\newcommand{\Renyi}{R{\'e}nyi~}
\newcommand{\Hell}{H}

\renewcommand{\d}{\textnormal{d}}

\newcommand{\Id}{{\mathds{1}}}
\DeclareMathAccent{\widehat}{0}{mathx}{"70}
\DeclareMathAccent{\widecheck}{0}{mathx}{"71}

\title{Some properties and applications of the new quantum $f$-divergences}

\author[1,2]{Salman Beigi}
\affil[1]{\footnotesize \it School of Mathematics, Institute for Research in Fundamental Sciences (IPM), P.O. Box 19395-5746, Tehran, Iran}
\affil[2]{\it Centre for Quantum Technologies, National University of Singapore, Singapore}
\author[3]{Christoph Hirche}
\affil[3]{\it Institute for Information Processing (tnt/L3S), Leibniz Universit\"at Hannover, Germany}
\author[2,4]{Marco Tomamichel}
\affil[4]{\it Department of Electrical and Computer Engineering, National University of Singapore, Singapore}

\begin{document}

\maketitle

\begin{abstract}
Recently, a new definition for quantum $f$-divergences was introduced based on an integral representation. These divergences have shown remarkable properties, for example when investigating contraction coefficients under noisy channels. At the same time, many properties well known for other definitions have remained elusive for the new quantum $f$-divergence because of its unusual representation. 
In this work, we investigate alternative ways of expressing these quantum $f$-divergences. We leverage these expressions to prove new properties of these $f$-divergences and demonstrate some applications. In particular, we give a new proof of the achievability of the quantum Chernoff bound by establishing a strengthening of an inequality by Audenaert et al. We 
also establish inequalities between some previously known \Renyi divergences and the new \Renyi divergence. We further investigate some monotonicity and convexity properties of the new $f$-divergences, and prove inequalities between these divergences for various functions.  
\end{abstract}

\tableofcontents

\section{Introduction}
Divergences are the central measures in quantum information theory. They provide a tool for quantifying resources such as information and entanglement. Understanding the properties of these divergences is the foundation on which many applications stand. However, there is not a single divergence, but many that find distinct applications. Hence, a natural goal is to find families of divergences with desirable common properties that unify the study of these quantities. 
One such family is given by $f$-divergences. In recent work~\cite{hirche2023quantum} a new family of quantum $f$-divergences was defined, for any $f\in\cF$, as
\begin{align}
    D_f(\rho\|\sigma) := \int_1^\infty \Big(f''(\gamma) E_\gamma(\rho\|\sigma) + \gamma^{-3} f''(\gamma^{-1})E_\gamma(\sigma\|\rho)\Big)\, \d \gamma. \label{Eq:f-divergence}
\end{align}
Here, $\cF$ denotes the set of functions $f:(0,\infty)\rightarrow\RR$ that are convex and twice differentiable with $f(1)=0$ and 
\begin{align}
E_\gamma(\rho\|\sigma) = \tr (\rho-\gamma \sigma)_+ ,
\end{align}
is the quantum Hockey-Stick divergence.\footnote{For a self-adjoint operator $X$, we denote by $\tr (X)_+$ the sum of positive eigenvalues of $X$.}
This $f$-divergence includes familiar special cases such as the Umegaki relative entropy for $f(\gamma)=\gamma\log(\gamma)$, but it also defines a new Hellinger divergence $H_\alpha(\rho\|\sigma)$ for $f_\alpha(\gamma)=\frac{\gamma^\alpha-1}{\alpha-1}$ that does not seem to correspond to any previous variant. These Hellinger divergences are of particular interest because they are in one-to-one correspondence with another important family of divergences, \Renyi divergences. The connection is given by 
\begin{align}
    D_\alpha(\rho\|\sigma) = \frac1{\alpha-1}\log\left(1+(\alpha-1)H_\alpha(\rho\|\sigma)\right). 
\end{align}

The defining expression for the $f$-divergence in terms of an integral over elementary Hockey-Stick divergences has many advantages when it comes to proving some of its properties. In particular, new contraction properties of these divergences under noisy channel were established in~\cite{hirche2023quantum}, see also~\cite{nuradha2024contraction} for additional recent results.
This expression also naturally connects these divergences to the concept of quantum differential privacy~\cite{hirche2022quantum}, leading, for example, to new bounds on problem of 
private hypothesis testing~\cite{cheng2024sample,nuradha2024contraction}. Besides those applications, these new divergences are also of further theoretical interest. In particular, the resulting \Renyi divergence regularizes to the most common definitions of quantum \Renyi divergences~\cite{hirche2023quantum}, namely the Petz \Renyi divergence for $0<\alpha<1$ and the sandwiched \Renyi divergence for $\alpha>1$. 

However, as might be expected, this representation has not only advantages. Certain properties that are well known for other $f$-divergences, for example the additivity of the Umegaki relative entropy under tensor products, are surprisingly hard to prove using the integral representation~\eqref{Eq:f-divergence}. In this work, we study alternative representations of the new $f$-divergences, and prove several new properties for them and provide some applications. 

One of the main applications of the quantum \Renyi divergences is that they characterize the optimal error exponent in several hypothesis testing scenarios. In particular, the Petz \Renyi divergence characterizes the error exponent in the symmetric quantum hypothesis testing problem, called the Chernoff exponent. 
Here, we give an alternative achievability proof of the Chernoff bound by replacing an inequality by Audenaert et al. which was used in the original proof~\cite{audenaert2007discriminating}, with the stronger result that for any $0\leq \alpha\leq 1$, 
    \begin{align}
        (1-\alpha) \Hell_\alpha(\rho\|\sigma) \leq \tr (\rho- \ \sigma)_+. \label{Eq:stongerAud}
    \end{align}
Compared to the inequality by Audenaert et al., this result has a remarkably simple proof, yet the proof of the fact that it implies the former is far from obvious. 
To that end, we need to compare the new quantum Hellinger divergence with the previous ones. This again turns out to be complicated due to the fundamentally different representation of these Hellinger divergences. For this reason, we aim to find other representations that more closely resemble the typical expressions for Hellinger divergences in terms of a trace of some operator. We call such expressions \textit{trace representations}. 

Here are some trace representations proven in our work. First, for all integer $\alpha\geq2$, we prove
\begin{align}
    H_\alpha(\rho\|\sigma)&=  \int_0^\infty \tr \left[ \left( \rho(\sigma+s\Id)^{-1} \right)^\alpha \right] \d s
    - \frac{1}{\alpha-1}, \\
        D_\alpha(\rho\|\sigma) &= \frac{1}{\alpha-1}\log \left( (\alpha-1) \int_0^\infty \tr\left[ \left( \rho(\sigma+s\Id)^{-1} \right)^\alpha \right] \d s\right). 
\end{align}
Second, for $0<\alpha<1$, we find that 
    \begin{align}
        H_\alpha(\rho\|\sigma) 
        &= \frac{\alpha}{(1-\alpha)} \int_0^\infty  
        \left(   \tr\Big[\big(\sigma^{\frac12}(\rho+r\Id)^{-1}\sigma^{\frac12}\big)^{1-\alpha}\Big]  -(1+r)^{\alpha-1}  \right)  \d r, 
    \end{align}
or, using an alternative expression, 
\begin{align}
        D_\alpha(\rho\|\sigma) 
        &= \frac{1}{\alpha-1}\log\left( \frac{\sin{\alpha\pi}}{-\pi} \int_0^\infty 
        \int_0^\infty t^{\alpha} \tr \left[ (\sigma(\rho+t\sigma+r\Id)^{-1})^2\right]\, \d r\,  \d t\right).
\end{align}
As consequences of the above expressions we show that 
for $0<\alpha<1$, we have 
    \begin{align}
        H_\alpha(\rho\|\sigma) &\geq \widebar H_\alpha(\rho\|\sigma), \label{Eq:HgeqPetzH}\\
        D_\alpha(\rho\|\sigma) &\geq \widebar D_\alpha(\rho\|\sigma)
    \end{align}
and for all integers $\alpha\geq2$, we have
    \begin{align}
        H_\alpha(\rho\|\sigma) \leq \widetilde H_\alpha(\rho\|\sigma), \\ 
        D_\alpha(\rho\|\sigma) \leq \widetilde D_\alpha(\rho\|\sigma). 
    \end{align}
Here, $\widebar D_\alpha(\rho\|\sigma)$ and $\widetilde D_\alpha(\rho\|\sigma)$ denote Petz and sandwiched \Renyi divergences respectively, and $\widebar H_\alpha(\rho\|\sigma)$ and $\widetilde H_\alpha(\rho\|\sigma)$ denote the corresponding Hellinger divergences.  
We conjecture that the latter relations indeed hold for all $\alpha>1$. We note that Equation~\eqref{Eq:HgeqPetzH} establishes our earlier claim that Equation~\eqref{Eq:stongerAud} is indeed a strengthening of the inequality by Audenaert et al. 

Besides the above, we explore two further types of alternative representations for the $f$-divergences. The first one allows us, in particular, to prove the monotonicity of $f$-divergences, meaning that for two functions $f,g\in\cF$ satisfying 
 $f(x)\geq g(x)$ we have
\begin{align}
    D_f(\rho\|\sigma) \geq D_g(\rho\|\sigma). 
\end{align}
While this seems like an elementary fact, it is not clear how to prove it from the initial representation as it only involves the second derivative of the function. As a consequence of the above inequality, we also show the previously elusive statement that for $\alpha\geq\beta$, 
\begin{align}
    H_\alpha(\rho\|\sigma) &\geq H_\beta(\rho\|\sigma), \\
    D_\alpha(\rho\|\sigma) &\geq D_\beta(\rho\|\sigma).
\end{align}

Regarding the other type of representations, we prove several new integral relationships between the $f$-divergences for various functions $f$. For example, we show that
    \begin{align}
        D_F(\rho\|\sigma) = \int_0^1 D_f(\rho\|(1-s)\rho+s\sigma) \frac{\d s}{s},
    \end{align}
where $F(\gamma)$ is a function related to $f(\gamma)$ by $F''(\gamma) =\frac{(\gamma-1) f'(\gamma)-f(\gamma)}{\gamma(\gamma-1)^2}$.   
These representations are essentially integral representations of the $f$-divergences with respect to related functions.  
These representations can be used to prove new inequalities between $f$-divergences. For example, 
    \begin{align}
        D\big(\rho\big\|(1-\lambda)\rho+\lambda\sigma\big) \leq \frac{\lambda^2}2 \left(\chi^2(\rho\|\sigma) + \chi^2(\sigma\|\rho)\right). 
    \end{align}
where $\chi^2(\rho\|\sigma)=H_2(\rho\|\sigma)$ is the chi-square divergence. 
A similar integral representation lets us prove reverse Pinsker type inequalities 
    \begin{align}
        D_f(\rho\|\sigma) \leq \frac{2\kappa^\uparrow_f(\rho,\sigma)}{\lambda_{\min}(\sigma)} E_1(\rho\|\sigma)^2,
    \end{align}
where for the definition of $\kappa^\uparrow_f(\rho, \sigma)$ we refer to Corollary~\ref{cor:reverse-Pinsker}.    
This inequality generalizes similar results in the classical case~\cite{george2024divergence}. Also, it is an improvement over some results in~\cite{hirche2023quantum} as it is quadratic in the trace distance on the right hand side. 

The remainder of this paper is structured as follows. In Section~\ref{Sec:Preliminaries} we give relevant notation and definitions, along with some preliminary 
results. In Section~\ref{Sec:New-vs-old} we present progress towards trace representations for the new $f$-divergences, leading to inequalities between various \Renyi divergences as well as the new achievability proof for the quantum Chernoff bound. In Section~\ref{Sec:Monotonicity}, we prove some monotonicity and convexity results, particularly for the Hellinger and \Renyi divergences. Finally, in Section~\ref{Sec:More-stuff} we relate $f$-divergences for different functions via integral representations and use those to derive new inequalities between various $f$-divergences.

\section{Definitions and preliminary results} \label{Sec:Preliminaries}

\subsection{Notation}

Throughout this work we consider only finite dimensional Hilbert spaces and linear operators acting on them. Amongst them, positive operators are hermitian operators that have only non-negative eigenvalues and quantum states are positive operators with unit trace. We use standard notation in the field of quantum information theory. Quantum states are denoted by lower-case greek letters $\rho$, $\sigma$, $\tau$, etc. For a general hermitian operator $X$ we denote by $X_+$ and $X_-$ the positive and negative parts of the eigen-decomposition of $X$, such that $X=X_+ - X_-$. We use $X^{-1}$ to denote the generalized (Penrose) inverse of $X$. We write $\rho\ll\sigma$ if the support of $\rho$ is contained in the support of $\sigma$ and we write $\rho\ll\gg\sigma$ if the supports of $\rho$ and $\sigma$ coincide. By $\lambda_{\min}(X)$ we denote the smallest and by $\lambda_{\max}(X)$ the largest eigenvalue of an operator $X$. 
Quantum channels are completely positive trace-preserving maps acting on the space of operators. We use $\log(x)$ to denote the natural logarithm. 

\subsection{$f$-divergences}

Let $\cF$ denote the set of functions $f:(0,\infty)\rightarrow\RR$ that are convex and twice differentiable with $f(1)=0$.
Following~\cite{hirche2023quantum}, $f$-divergences are defined, for any $f\in\cF$, as
\begin{align}
    D_f(\rho\|\sigma) := \int_1^\infty \Big(f''(\gamma) E_\gamma(\rho\|\sigma) + \gamma^{-3} f''(\gamma^{-1})E_\gamma(\sigma\|\rho)\Big)\, \d \gamma, 
\end{align}
where
\begin{align}
E_\gamma(\rho\|\sigma) = \tr (\rho-\gamma \sigma)_+ ,
\end{align}
is the quantum Hockey-Stick divergence. These include a number of notable special cases, particularly the Umegaki  relative entropy for $f(x)=x\log(x)$, 
\begin{align}
    D_{x\log(x)}(\rho\|\sigma) = D(\rho\|\sigma) = \tr[\rho(\log(\rho)-\log(\sigma))]. 
\end{align}
For $\alpha>0$, the Hellinger divergences for $f_\alpha(x)=\frac{x^\alpha-1}{\alpha-1}$ is, 
\begin{align}
    H_\alpha(\rho\|\sigma) \coloneqq D_{f_\alpha}(\rho\|\sigma). 
\end{align}
An important special case is the \textit{chi-square} divergence $\chi^2(\rho\|\sigma)=H_2(\rho\|\sigma)$. 
The quantum Le~Cam divergence~\cite{hirche2023quantum} is given by
\begin{align}
    LC_\lambda(\rho\|\sigma) \coloneqq D_{g_\lambda}(\rho\|\sigma),  \label{Eq:LeCam-def}
\end{align}
where $\lambda \in [0, 1]$ and $g_\lambda(x)=\lambda(1-\lambda)\frac{(x-1)^2}{\lambda x+(1-\lambda)}$, in analogy to the classical Le~Cam divergence~\cite{le2012asymptotic,gyorfi2001class}. 
According to~\cite{hirche2023quantum}, for $\lambda \in [0,1]$ and quantum states $\rho,\sigma$ we have 
    \begin{align}
        LC_\lambda(\rho\|\sigma) = \lambda \chi^2\big(\rho \big\|\lambda\rho+(1-\lambda)\sigma\big) + (1-\lambda) \chi^2\big(\sigma \big\|\lambda\rho+(1-\lambda)\sigma\big) \,. \label{Eq:LCasX2}
   \end{align}
As also observed in~\cite{hirche2023quantum} $g_\lambda(x) + \frac{x}{\lambda x+ (1-\lambda)}$ is an affine function of $x$, and as is clear from the definition, affine functions result in trivial $f$-divergences.
Hence,
\begin{align}\label{eq:LC-equiv-func}
    LC_\lambda(\rho\|\sigma) = D_{-\frac{x}{\lambda x+(1-\lambda)}}(\rho\| \sigma). 
\end{align}

These $f$-divergences inherit useful properties from the Hockey-Stick divergence. This includes positivity, monotonicity under quantum channels (data processing inequality) and convexity in the quantum states. For further details and properties we refer to~\cite{hirche2023quantum}.  

\subsection{Divergences and their derivatives}

In this work we often use the derivatives of the $f$-divergences, which we compute here. 
First, we recall the following result from~\cite{hirche2023quantum} that relates the second derivative of any $f$-divergence with the $\chi^2$-divergence. 
\begin{theorem}[{\cite[Theorem 2.8]{hirche2023quantum}}]\label{Thm:H2-is-limit-of-f}
    Let $f \in \cF$ and $\rho \ll\gg \sigma$. We have
    \begin{align}
        \frac{\partial^2}{\partial \lambda^2} D_f(\lambda\rho+(1-\lambda)\sigma\|\sigma) \Big|_{\lambda = 0} = f''(1) \,\chi^2(\rho\|\sigma). 
    \end{align}
    where $\chi^2(\rho\|\sigma) = H_2(\rho\|\sigma)$ is the Hellinger divergence of order $2$.
\end{theorem}
As a direct consequence of this theorem, we obtain an alternative expressions for $\chi^2(\rho\|\sigma)$ and $D_2(\rho\|\sigma)$. For this, let $f(x)= x \log x$ which satisfies $f''(1) = 1$ and thus
\begin{align}
    \chi^2(\rho\|\sigma) &= \frac{\partial^2}{\partial \lambda^2} D(\lambda\rho+(1-\lambda)\sigma\|\sigma) \Big|_{\lambda = 0} \\
    &= \int_0^\infty \mathrm{d}s \ \tr \left( (\sigma + s \Id)^{-1} (\rho - \sigma) (\sigma + s \Id)^{-1} (\rho - \sigma) \right) \label{Eq:OperatorX2}\\
    &= \int_0^\infty \mathrm{d}s \ \tr \left( \rho (\sigma + s \Id)^{-1} \rho (\sigma + s \Id)^{-1} \right) - 1 \,, 
\end{align}
where the second derivative has been calculated in~\cite{petz1998contraction}. This also implies
\begin{align}
    D_2(\rho\|\sigma) &= \log \int_0^\infty \mathrm{d}s \ \tr \left( \rho (\sigma + s \Id)^{-1} \rho (\sigma + s \Id)^{-1} \right) .
\end{align}
These results motivate our later approach in Section~\ref{Sec:New-vs-old}. 

Here, we give a stronger result that does not require the limit of $\lambda\rightarrow 0$ and covers higher derivatives. 
\begin{theorem}\label{Thm:Df-kth-D}
Let $k\geq 1$ be an integer and let $f \in \cF$ be $k+2$ times differentiable. For $\rho \ll\gg \sigma$ we have, 
    \begin{align}
     \frac{\partial^k}{\partial \lambda^k} D_f(\lambda\rho+(1-\lambda)\sigma\|\sigma) = D_{F_{k,\lambda}}(\rho\|\sigma), 
\end{align}
where 
\begin{align}\label{Eq:Fklambda}
    F_{k,\lambda}(\gamma)= (\gamma-1)^k f^{(k)}(\lambda\gamma+1-\lambda).
\end{align}
In particular,
\begin{align}
    \frac{\partial^k}{\partial \lambda^k} D_f(\lambda\rho+(1-\lambda)\sigma\|\sigma) \bigg|_{\lambda = 0} &= f^{(k)}(1) \, D_{(\gamma-1)^k}(\rho\|\sigma). \label{Eq:k-deriv-0}
\end{align}
\end{theorem}

\begin{proof}
Following the first steps of the proof of~\cite[Theorem 2.8]{hirche2023quantum}, we have,   
    \begin{align}
        D_f(\lambda\rho+(1-\lambda)\sigma\|\sigma) &= \lambda^2 \int_1^\infty  f''(\lambda\gamma+1-\lambda) E_{\gamma}(\rho\|\sigma)  +  \frac{1}{\gamma^3} f''\left(\lambda\gamma^{-1}+1-\lambda\right)  E_{\gamma}(\sigma\|\rho)  \,\d\gamma. \\
        &= \int_1^\infty  f_\lambda''(\gamma) E_{\gamma}(\rho\|\sigma)  +  \frac{1}{\gamma^3}  f_\lambda''\left(\gamma^{-1}\right)  E_{\gamma}(\sigma\|\rho)  \,\d\gamma,  
    \end{align}
    where $f_\lambda(\gamma) = f(\lambda\gamma+1-\lambda)$.
We note that the integration is indeed over a compact interval independent of $\lambda$. Hence, the corresponding derivatives can be moved into the integral. Furthermore, as the variables $\gamma$ and $\lambda$ are independent in $f_\lambda$, the respective derivatives can be exchanged and the desired result follows simply by observing 
\begin{align}
    \frac{\partial^k}{\partial \lambda^k} f(\lambda\gamma+1-\lambda) = (\gamma-1)^k f^{(k)}(\lambda\gamma+1-\lambda) = F_{k,\lambda}(\gamma). 
\end{align}
\end{proof}

The divergence $D_{(\gamma-1)^k}$ is often referred to as $\chi^k$-divergence~\cite{barnett2002approximating} where, not surprisingly, it appears in the Taylor expansion of classical $f$-divergences. Note, however, that $(\gamma-1)^k$ does not belong to $\cF$ in general since it is not necessarily convex on $(0,+\infty)$.\footnote{Hereafter, whenever we write $D_{F}(\rho\| \sigma)$ for some function $F$ we automatically assume that $F(\gamma)$ is twice differentiable and use Equation~\eqref{Eq:f-divergence} for the definition of $D_{F}(\rho\| \sigma)$, but we do not necessarily assume that $F$ is convex.}

Similar to the calculation yielding Equation~\eqref{Eq:OperatorX2}, we can specialize the above theorem to the function $x\log x$ and obtain trace expressions for the $\chi^k$-divergence. 
\begin{corollary}\label{cor:(gamma-1)^k-integral}
    For any integer $k\geq 2$ and $\rho \ll\gg \sigma$ we have, 
    \begin{align}
        D_{(\gamma-1)^k}(\rho\|\sigma) = (k-1) \int_0^\infty \tr\left( (\rho-\sigma)(\sigma+s\Id)^{-1} \right)^k \d s. 
    \end{align}
\end{corollary}
\begin{proof}
    A direct computation of the derivatives of the Umegaki relative entropy yields
    \begin{align}
        \frac{\partial^k}{\partial \lambda^k} D(\lambda\rho+(1-\lambda)\sigma\|\sigma) \Big|_{\lambda = 0} = (-1)^k(k-1)! \int_0^\infty \tr\left( (\rho-\sigma)(\sigma+s\Id)^{-1} \right)^k \d s.
    \end{align}
    Furthermore, for the function $f(x)=x\log x$ we have 
    \begin{align}
        f^{(k)}(1) = (-1)^k(k-2)!\,. 
    \end{align}
    Plugging both into Equation~\eqref{Eq:k-deriv-0}, we get 
    \begin{align}
        (-1)^k(k-1)! \int_0^\infty \tr\left( (\rho-\sigma)(\sigma+s\Id)^{-1} \right)^k \d s &= (-1)^k(k-2)! \, D_{(\gamma-1)^k}(\rho\|\sigma), 
    \end{align}
    from which the desired result follows. 
    
\end{proof}

Alternatively, we can express the $\chi^k$-divergence in terms of the Hellinger divergence.
\begin{lemma}\label{Lem:chi-k-Ha}
    We have for integer $k\geq 2$ and $\rho \ll\gg \sigma$,
    \begin{align}
        D_{(\gamma-1)^k}(\rho\|\sigma) = \sum_{\alpha=2}^k \binom{k}{\alpha} (-1)^{k-\alpha} (\alpha-1) H_\alpha(\rho\|\sigma). 
    \end{align}
\end{lemma}
\begin{proof}
    By the binomial formula,
    \begin{align}
        (\gamma-1)^k = \sum_{\alpha=0}^k \binom{k}{\alpha} (-1)^{k-\alpha} \gamma^\alpha. 
    \end{align}
    On the other hand, the Hellinger divergence is given by $f(\gamma)=\frac{\gamma^\alpha-1}{\alpha-1}$. Ignoring linear terms, this implies 
    \begin{align}
        D_{\gamma^\alpha}(\rho\|\sigma) = (\alpha-1)H_\alpha(\rho\|\sigma).
    \end{align}
    Combining both, we get the claim by noting again that linear functions lead to trivial divergences.
\end{proof}
Lemma~\ref{Lem:chi-k-Ha} together with Theorem~\ref{Thm:Df-kth-D} give the following special cases,
\begin{align}
        \frac{\partial^3}{\partial \lambda^3} D_f(\lambda\rho+(1-\lambda)\sigma\|\sigma) \Big|_{\lambda = 0} &= f^{(3)}(1) \, \big(2H_3(\rho\|\sigma) -3H_2(\rho\|\sigma) \big)  \\
        \frac{\partial^4}{\partial \lambda^4} D_f(\lambda\rho+(1-\lambda)\sigma\|\sigma) \Big|_{\lambda = 0} &= f^{(4)}(1) \, \big(3H_4(\rho\|\sigma)-8H_3(\rho\|\sigma)+6H_2(\rho\|\sigma)\big)  \\
        \frac{\partial^5}{\partial \lambda^5} D_f(\lambda\rho+(1-\lambda)\sigma\|\sigma) \Big|_{\lambda = 0} &= f^{(5)}(1) \, \big(4H_5(\rho\|\sigma)-15H_4(\rho\|\sigma)+20H_3(\rho\|\sigma)-10H_2(\rho\|\sigma)\big) .
\end{align}
Also, using this lemma together with Corollary~\ref{cor:(gamma-1)^k-integral}, we derive the following recursive expression for the Hellinger divergence with integer parameter $k$:
\begin{align}
    H_k(\rho\|\sigma) = \int_0^\infty \tr\left( (\rho-\sigma)(\sigma+s\Id)^{-1} \right)^k \d s + \sum_{\alpha=2}^{k-1} \binom{k}{\alpha} (-1)^{k-\alpha-1} \frac{(\alpha-1)}{(k-1)} H_\alpha(\rho\|\sigma). 
\end{align}
For instance, using Equation~\eqref{Eq:OperatorX2} for $H_2(\rho\| \sigma)=\chi^2(\rho\|\sigma)$ we find that
\begin{align}
    H_3(\rho\|\sigma) &=  \int_0^\infty \tr\left( (\rho-\sigma)(\sigma+s\Id)^{-1} \right)^3 \d s + \frac32 \int_0^\infty \tr\left( (\rho-\sigma)(\sigma+s\Id)^{-1} \right)^2 \d s\\
&= \int_0^\infty \tr\left( \rho(\sigma+s\Id)^{-1} \right)^3 \d s -\frac12,
\end{align}
where the second line follows by a straightforward computation. In the following 
section we will find a similar expression for Hellinger divergence $H_\alpha(\rho\|\sigma)$ for any integer $\alpha$.

\section{Inequalities relating the new and old \Renyi divergences and their applications}\label{Sec:New-vs-old}

\Renyi divergences are  common tools in information theory. This is especially because the classical \Renyi divergence gives the optimal error exponent and strong converse exponent in the discrimination of two probability distributions. When discriminating two quantum states, this role is taken by the Petz and the sandwiched \Renyi divergences. These are known to be equal to the regularization of the new \Renyi divergence 
in certain ranges of parameters~\cite{hirche2023quantum}. On the other  hand, the underlying Hockey-Stick divergence is naturally related to the error probability in state discrimination as we will discuss further later. In this section, we will establish inequalities between the different \Renyi divergences and discuss their applications, particularly to the state discrimination problem. To this end, we first give an alternative representation of the new \Renyi divergence that more closely resembles the traditional expressions. 

\subsection{Towards trace representations for the new $f$-divergences} 

In this section we use the tools developed in the previous section to derive alternative expressions for the Hellinger divergence $H_\alpha(\rho\| \sigma)$. 
In the following, we will give two approaches governing two different regions for the parameter $\alpha$. A complete solution for other values of $\alpha$ remains an open problem. 

\subsubsection{Hellinger divergence for integer $\alpha\geq 2$ }

Recall that the Hellinger divergence $H_{\alpha}(\rho\| \sigma)$ is associated with the function $f(\gamma)=\frac{\gamma^\alpha-1}{\alpha-1}$ and we have
\begin{align}
    \gamma^\alpha = \big(1+ (\gamma-1)\big)^\alpha =\sum_{n=0}^\infty {\alpha \choose n} (\gamma-1)^n =\sum_{n=0}^\infty \frac{\alpha\faf{n}}{n!} (\gamma-1)^n,
\end{align}
where $\alpha\faf{n}=\alpha(\alpha-1)\cdots(\alpha-n+1)$ is the falling factorial. This series is not convergent for general $\alpha$ and $\gamma > 2$. However, when $\alpha\geq 2$ is an integer, it is a finite sum and we may write the following series for the Hellinger divergence, 
\begin{align}
    H_\alpha(\rho\|\sigma) &= \frac{1}{\alpha-1} \sum_{n=2}^\infty \frac{\alpha\faf{n}}{n!} \, D_{(\gamma-1)^n}(\rho\|\sigma) \\
    &=  \frac{1}{\alpha-1} \sum_{n=2}^\infty \frac{\alpha\faf{n}}{n!} \, (n-1) \int_0^\infty \tr\left( (\rho-\sigma)(\sigma+s\Id)^{-1} \right)^n \d s,  \label{eq:Taylor-exp-Hellinger-integer}
\end{align}
where the second equality follows from Lemma~\ref{Lem:chi-k-Ha}. 
 This can be used to obtain an trace representation for $H_\alpha(\rho\| \sigma)$. To this end, we first compute an equivalent expression for the integral on the right hand side.

\begin{lemma}\label{Lem:Aux-exp-conv}
For any integer $n\geq2$, we have,
    \begin{align}
    \int_0^\infty \tr\left( (\rho-\sigma)(\sigma+s\Id)^{-1} \right)^n \d s = \sum_{k=2}^n (-1)^{n-k} \frac{n}{k} {n-2\choose k-2}  \int_0^\infty \tr\left( \rho(\sigma+s\Id)^{-1} \right)^k \d s + (-1)^{n-1}.
\end{align}
\end{lemma}
\begin{proof}
    We expand,
    \begin{align}
        \left( (\rho-\sigma)(\sigma+s\Id)^{-1} \right)^n  =  \underbrace{\big(\rho(\sigma+s\Id)^{-1} -\sigma(\sigma+s\Id)^{-1}\big)\cdots \big(\rho(\sigma+s\Id)^{-1} -\sigma(\sigma+s\Id)^{-1}\big)}_{n \text{ factors}},
    \end{align}
   to obtain terms of the form 
   \begin{align}
   \big(-\sigma(\sigma+s\Id)^{-1}\big)^{j_1}\big(\rho(\sigma+s\Id)\big) \big(-\sigma(\sigma+s\Id)^{-1}\big)^{j_2} \big(\rho(\sigma+s\Id)\big)\cdots\big(-\sigma(\sigma+s\Id)^{-1}\big)^{j_{k+1}},
   \end{align}
where $j_1, \dots, j_{k+1}$ are non-negative integers satisfying $j_1+\cdots +j_{k+1} = n-k$. We find that
    \begin{align}
        &\tr\left[\left( (\rho-\sigma)(\sigma+s\Id)^{-1} \right)^n \right] - \tr\left[(-\sigma(\sigma+s\Id)^{-1})^n\right] \\
        &= \sum_{k=1}^n \sum_{\substack{j_1,\dots,j_{k+1} \\ j_1+\cdots+j_{k+1}=n-k}} \tr\left[ (-\sigma(\sigma+s\Id)^{-1})^{j_1}(\rho(\sigma+s\Id))\cdots(-\sigma(\sigma+s\Id)^{-1})^{j_{k+1}} \right] \\
        &= \sum_{k=1}^n \sum_{\substack{j_1,\dots,j_k \\ j_1+\cdots+j_k=n-k}} (j_k+1) \tr\left[ (\rho(\sigma+s\Id))(-\sigma(\sigma+s\Id)^{-1})^{j_1}
        \cdots(-\sigma(\sigma+s\Id)^{-1})^{j_k} \right] \\
        &= \sum_{k=1}^n \sum_{\substack{j_1,\dots,j_k \\ j_1+\cdots+j_k=n-k}} \frac{n}{k} \tr\left[ (\rho(\sigma+s\Id))(-\sigma(\sigma+s\Id)^{-1})^{j_1}
        \cdots(-\sigma(\sigma+s\Id)^{-1})^{j_k} \right]. \label{eq:proof-lem:Aux-exp-conv-3.10}
    \end{align}
    where the penultimate line follows from the cyclic property of the trace function and the last line follows by averaging $j_k+1$ over cyclic shifts $(j_1, \dots, j_k)\mapsto (j_2, j_3, \dots, j_k, j_1)$.
    To continue, we note that
\begin{align}
        \frac{\d^{n-k}}{\d r^{n-k}}  (\rho(r\sigma+s\Id)^{-1})^k  
        = \sum_{\substack{j_1,\dots,j_k \\ j_1+\dots+j_k=n-k}} \binom{n-k}{j_1,\dots,j_k}   \frac{\d^{j_1}}{\d r^{j_1}}\big(\rho(r\sigma+s\Id)^{-1}\big)
        \cdots \frac{\d^{j_k}}{\d r^{j_k}}\big(\rho(r\sigma+s\Id)^{-1}\big)  . 
    \end{align}
Also, for any $j$ we have
\begin{align}
    \frac{\d^j}{\d r^j} \big(\rho(r\sigma+s\Id)^{-1}\big) = j! \big(\rho(r\sigma+s\Id)^{-1}\big) \big(-\sigma(r\sigma+s\Id)^{-1}\big)^j .
\end{align}
Therefore,    
    \begin{align}
        &\frac{\d^{n-k}}{\d r^{n-k}} \tr\left[ (\rho(r\sigma+s\Id)^{-1})^k \right] \\
        &= (n-k)! \sum_{\substack{j_1,\dots,j_k \\ j_1+\dots+j_k=n-k}}  \tr\left[ \big(\rho(r\sigma+s\Id)\big)\big(-\sigma(r\sigma+s\Id)^{-1}\big)^{j_1}
        \cdots\big(-\sigma(r\sigma+s\Id)^{-1}\big)^{j_k} \right]. 
    \end{align}
    Comparing to Equation~\eqref{eq:proof-lem:Aux-exp-conv-3.10} we arrive at 
    \begin{align}
        \tr\left[\left( (\rho-\sigma)(\sigma+s\Id)^{-1} \right)^n \right] 
        &= \tr\left[(-\sigma(\sigma+s\Id)^{-1})^n\right] + n \tr\left[ \rho(\sigma+s\Id)^{-n}(-\sigma)^{n-1}\right] \\
        &\quad\, + \sum_{k=2}^n \frac{n}{k (n-k)!} \frac{\d^{n-k}}{\d r^{n-k}} \tr\left[ (\rho(r\sigma+s\Id)^{-1})^k \right]\bigg|_{r = 1}.  
    \end{align}
    Computing the integrals, we have 
    \begin{align}
        \int_0^\infty \tr\left[(-\sigma(\sigma+s\Id)^{-1})^n\right] \d s = \frac{(-1)^n}{n-1}, 
    \end{align} 
    and
    \begin{align}
        \int_0^\infty \tr\left[ \rho(\sigma+s\Id)^{-n}(-\sigma)^{n-1}\right] \d s = \frac{(-1)^{n-1}}{n-1}.
    \end{align}
    With these we have,
    \begin{align}
        &\int_0^\infty \tr\left[\left( (\rho-\sigma)(\sigma+s\Id)^{-1} \right)^n \right] \d s \\ 
        &= \frac{(-1)^n}{n-1} + n \frac{(-1)^{n-1}}{n-1} + \int_0^\infty \sum_{k=2}^n \frac{n}{k (n-k)!} \frac{\d^{n-k}}{\d r^{n-k}} \tr\left[ (\rho(r\sigma+s\Id)^{-1})^k \right]\bigg|_{r = 1} \d s \\ 
        &= (-1)^{n-1} +  \sum_{k=2}^n \frac{n}{k (n-k)!} \frac{\d^{n-k}}{\d r^{n-k}} \int_0^\infty\tr\left[ (\rho(r\sigma+s\Id)^{-1})^k \right]\d s \bigg|_{r = 1} \\ 
        &= (-1)^{n-1} +  \sum_{k=2}^n \frac{n}{k (n-k)!} \frac{\d^{n-k}}{\d r^{n-k}} \int_0^\infty r \tr\left[ (\rho(r\sigma+rt\Id)^{-1})^k \right]\d t \bigg|_{r = 1} \\ 
        &= (-1)^{n-1} +  \sum_{k=2}^n \frac{n}{k (n-k)!} \frac{\d^{n-k}}{\d r^{n-k}} \int_0^\infty r^{-k+1} \tr\left[ (\rho(\sigma+t\Id)^{-1})^k \right]\d t \bigg|_{r = 1},
    \end{align}
    where in the penultimate line we have substituted $s=rt$. We can now easily evaluate the derivative to get,
    \begin{align}
        &\int_0^\infty \tr\left[\left( (\rho-\sigma)(\sigma+s\Id)^{-1} \right)^n \right] \d s \\ 
        &= (-1)^{n-1} +  \sum_{k=2}^n \frac{n}{k (n-k)!} \frac{(n-2)!}{(k-2)!} (-1)^{n-k} \int_0^\infty  \tr\left[ (\rho(\sigma+t\Id)^{-1})^k \right]\d t,
    \end{align}   
    which is equivalent to the desired equality.
    
\end{proof}

We can now state our result regarding the Hellinger divergence.

\begin{theorem}\label{Thm:Ha-2int}
    For integer $\alpha\geq2$, we have
    \begin{align}
        H_\alpha(\rho\|\sigma)&=  \int_0^\infty \tr\left( \rho(\sigma+s\Id)^{-1} \right)^\alpha \d s
    - \frac{1}{\alpha-1}  .
    \end{align}
\end{theorem}

\begin{proof}
    Using Equation~\eqref{eq:Taylor-exp-Hellinger-integer} and Lemma~\ref{Lem:Aux-exp-conv} we find that
    \begin{align}
    &H_\alpha(\rho\|\sigma)\\
    &=  \frac{1}{\alpha-1} \sum_{n=2}^\alpha \frac{\alpha\faf{n}}{n!} \, (n-1) \left[\sum_{k=2}^n (-1)^{n-k} \frac{n}{k}{n-2 \choose k-2} \int_0^\infty \tr\left( \rho(\sigma+s\Id)^{-1} \right)^k \d s + (-1)^{n-1}\right] \\
    &=  \frac{1}{\alpha-1} \sum_{n=2}^\alpha \frac{\alpha\faf{n}}{n!} \,  \left[\sum_{k=2}^n (-1)^{n-k} \frac{n!}{k (n-k)! (k-2)!} \int_0^\infty \tr\left( \rho(\sigma+s\Id)^{-1} \right)^k \d s\right] 
    - \frac{1}{\alpha-1} \sum_{n=2}^\alpha \frac{\alpha\faf{n}}{n!} \, (n-1) (-1)^{n} .\label{Eq:Ha-Taylor-1}
\end{align}
We start by evaluating the first term in the above equation. Changing the order of summations yields
\begin{align}
    &\frac{1}{\alpha-1} \sum_{n=2}^\alpha \frac{\alpha\faf{n}}{n!} \,  \left[\sum_{k=2}^n (-1)^{n-k} \frac{n!}{k (n-k)! (k-2)!} \int_0^\infty \tr\left( \rho(\sigma+s\Id)^{-1} \right)^k \d s\right] \\
   &=\frac{1}{\alpha-1} \sum_{k=2}^\alpha \frac1{k} \left[\int_0^\infty \tr\left( \rho(\sigma+s\Id)^{-1} \right)^k \d s\right]
   \sum_{n=k}^\alpha \frac{\alpha\faf{n}}{n!} \,   (-1)^{n-k} \frac{n!}{(n-k)! (k-2)!} .\label{Eq:First-term-1}
\end{align}
Next, considering only the inner sum, we have
\begin{align}
    \sum_{n=k}^\alpha \frac{\alpha\faf{n}}{n!} \,   (-1)^{n-k} \frac{n!}{(n-k)! (k-2)!}     &=\sum_{n=0}^{\alpha-k} \frac{\alpha\faf{n+k}}{n! (k-2)!} \, (-1)^{n} \\
    &=\sum_{n=0}^{\alpha-k} \frac{\alpha\faf{k}(\alpha-k)\faf{n}}{n! (k-2)!} \, (-1)^{n} \\
    &=\frac{\alpha\faf{k}}{(k-2)!} \sum_{n=0}^{\alpha-k} \binom{\alpha-k}{n} \, (-1)^{n} \\
    &=\begin{cases}
        \frac{\alpha\faf{k}}{(k-2)!} & k=\alpha \\
        0 & \text{ otherwise.}
    \end{cases}   \label{Eq:inner-1}
\end{align}
Invoking this in Equation~\eqref{Eq:Ha-Taylor-1} we find that 
\begin{align}
H_\alpha(\rho\| \sigma)& =  \int_0^\infty \tr\left( \rho(\sigma+s\Id)^{-1} \right)^\alpha \d s - \frac{1}{\alpha-1} \sum_{n=2}^\alpha \frac{\alpha\faf{n}}{n!} \, (n-1) (-1)^{n} .
\end{align}
We finish by evaluating the second term as follows:
\begin{align}
    \frac{1}{\alpha-1} \sum_{n=2}^\alpha \frac{\alpha\faf{n}}{n!} \, (n-1) (-1)^{n} 
    &= \alpha \sum_{n=0}^{\alpha-2} \frac{1}{(n+2)} \binom{\alpha-2}{n} \, (-1)^{n+2}\\
    &= \alpha \sum_{n=0}^{\alpha-2}  \binom{\alpha-2}{n} \, \int_0^{-1} x^{n+1}\d x\\
    &= \alpha \int_0^{-1} x \sum_{n=0}^{\alpha-2}  \binom{\alpha-2}{n} \,  x^{n}\d x\\
    & = \alpha \int_0^{-1} x(x+1)^{\alpha-2}\d x\\
    & = -\frac{\alpha}{\alpha-1} \int_0^{-1} (x+1)^{\alpha-1}\d x\\
    & = \frac{1}{\alpha-1},
\end{align}
where in the penultimate line we apply integration by parts. 
\end{proof}
This theorem also implies that for any integer $\alpha\geq 2$, 
\begin{align}
    D_\alpha(\rho\|\sigma) = \frac{1}{\alpha-1}\log \left[(\alpha-1) \int_0^\infty \tr\left( \rho(\sigma+s\Id)^{-1} \right)^\alpha \d s\right].\label{Eq:trace-exp-Da} 
\end{align}

We conjecture that the above result indeed holds for all $\alpha>1$. 
\begin{conjecture}\label{Conj:Ha-2int}
    For any $\alpha>1$, we have \begin{align}
       H_\alpha(\rho\|\sigma)&=  \int_0^\infty \tr\left( \rho(\sigma+s\Id)^{-1} \right)^\alpha \d s
    - \frac{1}{\alpha-1}  .\label{Eq:Conjecture}
    \end{align}
\end{conjecture}
In the following we provide some evidences in the favour of this conjecture.

\begin{proposition}\label{Prop:Ha-2int-pure}
    If $\rho$ is a pure quantum state and $\rho\ll\sigma$, then for any $\alpha>1$, we have
    \begin{align}
       H_\alpha(\rho\|\sigma)&=  \int_0^\infty \tr\left( \rho(\sigma+s\Id)^{-1} \right)^\alpha \d s
    - \frac{1}{\alpha-1}  .
    \end{align}
\end{proposition}
\begin{proof}
    Since $\rho$ is pure, $\rho-x\sigma$ has at most one positive eigenvalue for $x\geq0$. Indeed, for  $x\in[0, \beta]$ where $\beta=2^{D_{\max}(\rho\|\sigma)}$, $\rho-x\sigma$ has exactly one non-negative eigenvalue, which we denote by $\lambda(x)$, and it is negative definite for $x>\beta$.
    We note that $\lambda(0)=1$ and $\lambda(\beta)=0$. Moreover, by the uniqueness of $\lambda(x)$ and employing the implicit function theorem, $\lambda(x)$ is a smooth function of $x$. Using an expression from~\cite{hirche2023quantum} with $\tilde E_\gamma(\rho\|\sigma) = E_\gamma(\rho\|\sigma) -(1-\gamma)_+$, we have,
    \begin{align}
    D_f(\rho\|\sigma) &= \int_0^\infty f''(\gamma) \tilde E_\gamma(\rho\|\sigma) \d\gamma \\
    &= \int_0^\beta f''(\gamma) \lambda(\gamma) \d\gamma - \int_0^1 f''(\gamma) (1-x) \d\gamma \\
    &= \int_0^\beta f''(\gamma) \lambda(\gamma) \d\gamma - (-f'(0)+f(1)-f(0)). 
\end{align}
This implies, 
\begin{align}
    H_\alpha(\rho\|\sigma) = \alpha \int_0^\beta \gamma^{\alpha-2} \lambda(\gamma) \d\gamma - \frac{1}{\alpha-1}. 
\end{align}
Since $\lambda(x)$ is an eigenvalue of $\rho-x\sigma$, for $x\in[0,\beta]$ we have $\det(\rho-x\sigma-\lambda(x))= 
     \det\left(\rho-x\left(\sigma-\frac{\lambda(x)}{x}\right)\right) =0$ or equivalently
\begin{align} 
    \det\left( \left(\sigma+\frac{\lambda(x)}{x}\right)^{-\frac12} \rho \left(\sigma+\frac{\lambda(x)}{x}\right)^{-\frac12} - x \right) =0. 
\end{align}
This means that $x$ is an eigenvalue of 
\begin{align}    \left(\sigma+\frac{\lambda(x)}{x}\right)^{-\frac12} \rho \left(\sigma+\frac{\lambda(x)}{x}\right)^{-\frac12}
\end{align}
and since it is a rank-one matrix, $x$ is the only non-zero eigenvalue. Hence, recalling that $\lambda(0)=1$ and $\lambda(\beta)=0$, applying the change of variable $s=\lambda(x)/x$, we find that 
\begin{align}
    \int_0^\infty \tr\left[\left( \rho\left(\sigma+s\Id\right)^{-1} \right)^\alpha\right] \d s
    - \frac{1}{\alpha-1} &= - \int_0^\beta \tr\left[\left( \rho\left(\sigma+\frac{\lambda(x)}{x}\Id\right)^{-1} \right)^\alpha\right] \left(\frac{\lambda(x)}{x}\right)' \d x
    - \frac{1}{\alpha-1} \\
    &= - \int_0^\beta x^\alpha \left(\frac{\lambda(x)}{x}\right)' \d x
    - \frac{1}{\alpha-1} \\
    &= - \beta^{\alpha-1} \lambda(\beta) + \alpha \int_0^\beta x^{\alpha-2} \lambda(x)  \d x
    - \frac{1}{\alpha-1} \\
    &=  \alpha \int_0^\beta x^{\alpha-2} \lambda(x)  \d x
    - \frac{1}{\alpha-1}, 
\end{align}
where in the third line we apply integration by parts, and in the last line we use $\lambda(\beta)=0$.
This matches the earlier expression for $H_\alpha(\rho\| \sigma)$ and hence proves the claim. 
\end{proof}

We have also numerically verified the conjecture for some simple examples. One such example is illustrated in Figure~\ref{Fig:Conjecture}. 
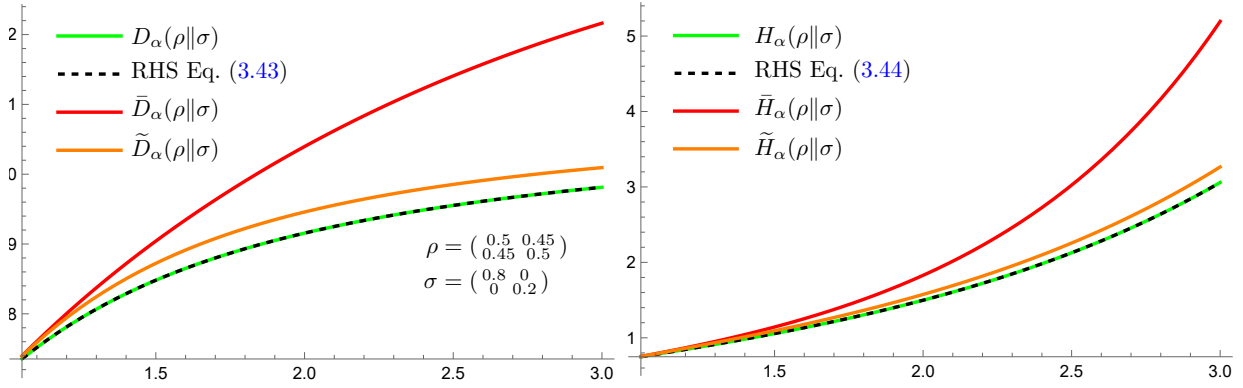
\begin{figure}[t]
\centering
\begin{minipage}{.49\textwidth}
\scalebox{0.88}{\begin{tikzpicture}
    \node[anchor=south west,inner sep=0] (image) at (0,0) {\includegraphics[width=1.136\textwidth]{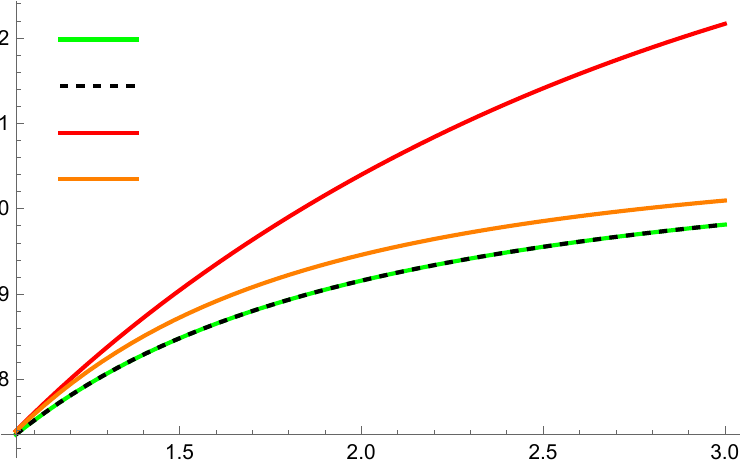}};
    \begin{scope}[x={(image.south east)},y={(image.north west)}]
        \node[] at (0.33,0.76){$\begin{aligned}
        &D_\alpha(\rho\|\sigma) \\
        &\text{RHS Eq.~\eqref{Eq:trace-exp-Da}} \\
        &\widebar D_\alpha(\rho\|\sigma) \\
        &\widetilde D_\alpha(\rho\|\sigma) 
        \end{aligned}$};
        \node[] at (0.8,0.31){$\begin{aligned}
        \rho&=\left( \begin{smallmatrix} 0.5 & 0.45  \\ 0.45 & 0.5  \end{smallmatrix}\right) \\
        \sigma&=\left( \begin{smallmatrix} 0.8 & 0  \\ 0 & 0.2  \end{smallmatrix}\right) 
        \end{aligned}$};
    \end{scope}
\end{tikzpicture}}
\end{minipage}
\begin{minipage}{.49\textwidth}
\scalebox{0.88}{\begin{tikzpicture}
    \node[anchor=south west,inner sep=0] (image) at (0,0) {\includegraphics[width=1.136\textwidth]{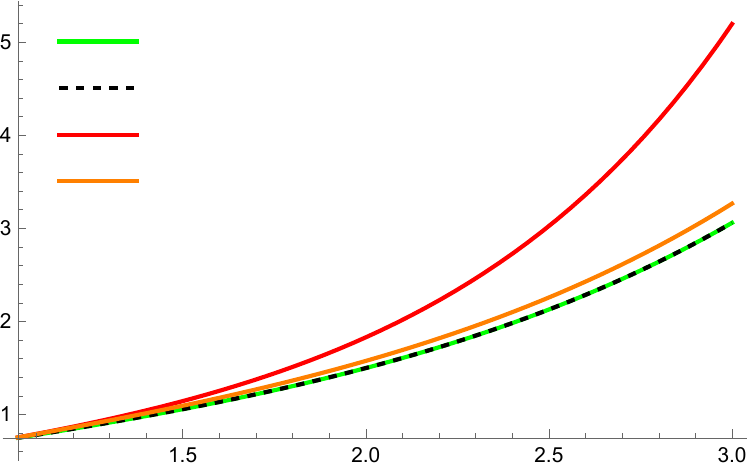}};
    \begin{scope}[x={(image.south east)},y={(image.north west)}]
        \node[] at (0.34,0.76){$\begin{aligned}
        &\Hell_\alpha(\rho\|\sigma) \\
         &\text{RHS Eq.~\eqref{Eq:Conjecture}} \\
        &\widebar \Hell_\alpha(\rho\|\sigma) \\
        &\widetilde \Hell_\alpha(\rho\|\sigma) 
        \end{aligned}$};
    \end{scope}
\end{tikzpicture}}
\end{minipage}
\caption{\label{Fig:Conjecture} Plot of the different \Renyi divergences (left) and Hellinger divergences (right) defined in the main text over $\alpha$ for the states specified above. For this and all other examples we tried, Conjecture~\ref{Conj:Ha-2int} appears to hold.}
\end{figure}
\subsubsection{Hellinger divergence for $0<\alpha<1$}

In this section, we use a different approach based on an integral representation of $x^\alpha$. For $0<\alpha<1$, we have
\begin{align}
    x^\alpha = \frac{\sin{\alpha\pi}}{\pi} \int_0^\infty t^{\alpha-1}\frac{x}{x+t} \d t. \label{Eq:xa-integral-rep}
\end{align}
We can use this to give the following representation of the Hellinger divergence.


\begin{proposition}\label{Prop:int-Ha-012}
For $0<\alpha<1$, we have 
    \begin{align}
        H_\alpha(\rho\|\sigma) 
        &= \frac{\sin{\alpha\pi}}{\pi(1-\alpha)} \int_0^\infty t^{\alpha}
        \left[ \int_0^\infty  \tr\left[\big(\sigma(\rho+t\sigma+r\Id)^{-1}\big)^2\right] \d r -\frac{1}{1+t}  \right]  \d t.
    \end{align}
\end{proposition}
\begin{proof}
    We start by substituting $\lambda=\frac1{t+1}$ in Equation~\eqref{Eq:xa-integral-rep}, leading to,
    \begin{align}
         x^\alpha = \frac{\sin{\alpha\pi}}{\pi} \int_0^1 \left(\frac{1-\lambda}{\lambda}\right)^{\alpha-1}\frac{\lambda x}{\lambda x+(1-\lambda)} \frac{\d \lambda}{\lambda^2}. \label{Eq:xa-integral-rep-2}   
    \end{align}
    This allows us express the Hellinger divergence as, 
    \begin{align}
        H_\alpha(\rho\|\sigma) &= \frac{\sin{\alpha\pi}}{\pi(1-\alpha)} \int_0^1 \left(\frac{1-\lambda}{\lambda}\right)^{\alpha-1}D_{\frac{- x}{\lambda x+(1-\lambda)}}(\rho\|\sigma)  \frac{\d \lambda}{\lambda} \\
        &= \frac{\sin{\alpha\pi}}{\pi(1-\alpha)} \int_0^1 \left(\frac{1-\lambda}{\lambda}\right)^{\alpha-1}LC_{\lambda}(\rho\|\sigma)  \frac{\d \lambda}{\lambda},
    \end{align}
    where in the second equality we use Equation~\eqref{eq:LC-equiv-func}. Continuing with Lemma~\ref{lem:LC-X2-rel}, we get
    \begin{align}
        H_\alpha(\rho\|\sigma) 
        &= \frac{\sin{\alpha\pi}}{\pi(1-\alpha)} \int_0^1 \left(\frac{1-\lambda}{\lambda}\right)^{\alpha}\chi^2\big(\sigma \big\|\lambda\rho+(1-\lambda)\sigma\big)  \frac{\d \lambda}{\lambda} \\
        &= \frac{\sin{\alpha\pi}}{\pi(1-\alpha)} \int_0^1 \left(\frac{1-\lambda}{\lambda}\right)^{\alpha}
        \left[ \int_0^\infty \tr\left[\big(\sigma(\lambda\rho+(1-\lambda)\sigma+s\Id)^{-1}\big)^2\right] \d s -1  \right]  \frac{\d \lambda}{\lambda}\,.
    \end{align}
    We can now reverse the previous substitution by setting $t=\frac{1-\lambda}{\lambda}$, 
    \begin{align}
        H_\alpha(\rho\|\sigma) 
        &= \frac{\sin{\alpha\pi}}{\pi(1-\alpha)} \int_0^\infty t^{\alpha}
        \left[ \int_0^\infty \tr\left[\left(\sigma\left(\frac{1}{1+t}\rho+\frac{t}{1+t}\sigma+s\Id\right)^{-1}\right)^2\right] \d s -1  \right]  \frac{\d t}{1+t} \\ 
        &= \frac{\sin{\alpha\pi}}{\pi(1-\alpha)} \int_0^\infty t^{\alpha}
        \left[ \int_0^\infty (1+t)^2 \tr\left[\big(\sigma(\rho+t\sigma+(1+t)s\Id)^{-1}\big)^2\right] \d s -1  \right]  \frac{\d t}{1+t} \\
        &= \frac{\sin{\alpha\pi}}{\pi(1-\alpha)} \int_0^\infty t^{\alpha}
        \left[ \int_0^\infty (1+t) \tr\left[\big(\sigma(\rho+t\sigma+r\Id)^{-1}\big)^2\right] \d r -1  \right]  \frac{\d t}{1+t},
    \end{align}
    where in the last equality we further substituted $r=(1+t)s$. 
\end{proof}
In the next step we explicitly compute the integral over $t$ to derive an expression resembling more closely that of traditional Hellinger divergences. 
\begin{proposition}\label{Prop:Ha01-2nd}
    For $0<\alpha<1$, we have
    \begin{align}
        H_\alpha(\rho\|\sigma) 
        &= \frac{\alpha}{(1-\alpha)} \int_0^\infty  
        \left[   \tr[(\sigma^{\frac12}(\rho+r\Id)^{-1}\sigma^{\frac12})^{1-\alpha}]  -(1+r)^{\alpha-1}  \right]  \d r.
    \end{align}
\end{proposition}
\begin{proof}
    We start from Proposition~\ref{Prop:int-Ha-012} by rewriting the trace, 
    \begin{align}
        H_\alpha(\rho\|\sigma) 
        &= \frac{\sin{\alpha\pi}}{\pi(1-\alpha)} \int_0^\infty t^{\alpha}
        \left[ \int_0^\infty  \tr\left[\big(\sigma(\rho+t\sigma+r\Id)^{-1}\big)^2\right] \d r -\frac{1}{1+t}  \right]  \d t \\
        &= \frac{\sin{\alpha\pi}}{\pi(1-\alpha)} \int_0^\infty t^{\alpha}
        \left[ \int_0^\infty  \tr\left[\big(\sigma^{-\frac12}\rho\sigma^{-\frac12}+t+r\sigma^{-1}\big)^{-2}\right] \d r -\frac{1}{1+t}  \right]  \d t \\
        &= \frac{\sin{\alpha\pi}}{\pi(1-\alpha)} \int_0^\infty \int_0^\infty t^{\alpha}
        \left[   \tr\left[\big(\sigma^{-\frac12}\rho\sigma^{-\frac12}+t+r\sigma^{-1}\big)^{-2}\right]  -\frac{1}{(1+t+r)^2}  \right]  \d r\,\d t, 
    \end{align}
    where in the last line we use
    \begin{align}
        \int_0^\infty \frac{1}{(1+t+r)^2} \d r = \frac{1}{1+t}. 
    \end{align}
    Now we consider again the integral representation in Equation~\eqref{Eq:xa-integral-rep} and take the derivative on both sides to get
    \begin{align}
        \alpha x^{\alpha-1} = \frac{\sin{\alpha\pi}}{\pi} \int_0^\infty \frac{t^{\alpha}}{(x+t)^2} \d t. \label{Eq:xa-integral-rep-D}
\end{align}
    This has now a form that we can apply to our previous derivation, 
    \begin{align}
        H_\alpha(\rho\|\sigma) 
        &= \frac{\alpha}{(1-\alpha)} \int_0^\infty  
        \left[   \tr\left[\big(\sigma^{-\frac12}\rho\sigma^{-\frac12}+r\sigma^{-1}\big)^{\alpha-1}\right]  -(1+r)^{\alpha-1}  \right]  \d r \\
        &= \frac{\alpha}{(1-\alpha)} \int_0^\infty  
        \left[   \tr\left[\big(\sigma^{\frac12}(\rho+r\Id)^{-1}\sigma^{\frac12}\big)^{1-\alpha}\right]  -(1+r)^{\alpha-1}  \right]  \d r .
    \end{align}
\end{proof}

\subsubsection{Other divergences}

In general, one would expect to find trace expressions  for $f$-divergences, such as the ones above, for any suitable function $f$. However, for now this remains an open problem. Here, we briefly give some additional special cases were we can find trace expressions as a consequence of Theorem~\ref{Thm:Df-kth-D}. We start with the Le~Cam divergence corresponding to the function $g_\lambda(x) =\lambda(1-\lambda)\frac{(x-1)^2}{\lambda x+(1-\lambda)}$.

\begin{lemma}\label{lem:LC-X2-rel}
   We have,
    \begin{align}
        LC_\lambda(\rho\|\sigma) &= \frac{\lambda}{1-\lambda} \chi^2\big(\rho \big\|\lambda\rho+(1-\lambda)\sigma\big) \\
        &= \frac{1-\lambda}{\lambda} \chi^2\big(\sigma \big\|\lambda\rho+(1-\lambda)\sigma\big) \\
        &= \lambda(1-\lambda) \int_0^\infty \tr\left[(\rho-\sigma)(\rho_\lambda+s\Id)^{-1}(\rho-\sigma)(\rho_\lambda+s\Id)^{-1} \right]\d s\\
        &=\frac{\partial^2}{\partial \lambda^2} D(\lambda\rho+(1-\lambda)\sigma\|\sigma)\,,
    \end{align}
    where $\rho_\lambda= \lambda\rho+(1-\lambda)\sigma$.
\end{lemma}
\begin{proof}
    The first three lines can be easily verified by combining Equation~\eqref{Eq:LCasX2} with the expression in Equation~\eqref{Eq:OperatorX2}. 
    The last line follows from Theorem~\ref{Thm:Df-kth-D} by considering $f(\gamma)=\gamma\log\gamma$, in which case $f''(\gamma)=\frac1\gamma$ and then $F_{2, \lambda}(\gamma)=\frac1{\lambda(1-\lambda)}g_\lambda(\gamma)$. 
\end{proof}

Also the first derivative gives an trace expression, and recovers a well known identity from the literature. 
\begin{corollary}
For $F_{1,\lambda}(\gamma)= (\gamma-1) f'(\lambda\gamma+1-\lambda)$ from Equation~\eqref{Eq:Fklambda}, we have, 
\begin{align}
      D_{F_{1,\lambda}}(\rho\|\sigma)&= \frac{\partial}{\partial \lambda} D(\lambda\rho+(1-\lambda)\sigma\|\sigma) \\
     & =\frac{1}{\lambda} \left[D(\rho_\lambda\|\sigma) + D(\sigma\|\rho_\lambda)\right],
\end{align}
with $\rho_\lambda=\lambda\rho+(1-\lambda)\sigma$.
\end{corollary}
\begin{proof}
    The first equality is simply Theorem~\ref{Thm:Df-kth-D} for $k=1$. Observe that for $f(\gamma)=\gamma\log\gamma$, $F_{1,\lambda}(\gamma)=(\gamma-1) f'(\lambda\gamma+1-\lambda)=(\gamma-1)[1+\log(\lambda\gamma+1-\lambda)]$ and hence, with some manipulation, 
    \begin{align}
        F''_{1,\lambda}(\gamma)=\frac{\lambda}{\lambda\gamma+1-\lambda} + \frac{\lambda}{(\lambda\gamma+1-\lambda)^2}.
    \end{align}
    Clearly the first term here corresponds to $\lambda f''(\lambda\gamma+1-\lambda)=f''_\lambda(\gamma)/\lambda$ . For the second term, we additionally need the property that for $g(\gamma)=\gamma f(\gamma^{-1})$ we have $D_g(\rho\|\sigma)=D_f(\sigma\|\rho)$. Combining both we get the claimed result. 
\end{proof}
This derivative plays a crucial role for example in the context of log-Sobolev inequalities, see e.g.,~\cite{muller2016relative, kastoryano2013quantum}.


\subsection{Inequalities between \Renyi divergences}

In this section we use the representations developed above to prove relationships between the new \Renyi divergences and the previous definitions.
We start by comparing to the Petz \Renyi divergence~\cite{petz1985quasi,petz1986quasi}. 
\begin{proposition}\label{Prop:Ha-geq-tHa}
    For $0<\alpha<1$, we have 
    \begin{align}
        H_\alpha(\rho\|\sigma) &\geq \widebar H_\alpha(\rho\|\sigma), \\
        D_\alpha(\rho\|\sigma) &\geq \widebar D_\alpha(\rho\|\sigma) ,
    \end{align}
    where $\widebar H_\alpha(\rho\|\sigma) = \frac{1}{1-\alpha}\big(1-\tr\big[\rho^\alpha\sigma^{1-\alpha}\big]\big)$ and $\widebar D_\alpha(\rho\|\sigma) = \frac{1}{\alpha-1} \log \tr\big(\rho^\alpha\sigma^{1-\alpha}\big)$ is the Petz \Renyi divergence.
\end{proposition}
\begin{proof}
    Starting with Proposition~\ref{Prop:Ha01-2nd} and applying the Araki-Lieb-Thirring inequality~\cite{lieb76,araki1990inequality} we obtain 
    \begin{align}
        H_\alpha(\rho\|\sigma) 
        &= \frac{\alpha}{(1-\alpha)} \int_0^\infty  
        \left[   \tr[\big(\sigma^{\frac12}(\rho+r\Id)^{-1}\sigma^{\frac12}\big)^{1-\alpha}]  -(1+r)^{\alpha-1}  \right]  \d r \\
        &\geq \frac{\alpha}{(1-\alpha)} \int_0^\infty  
        \left[   \tr\big[\sigma^{1-\alpha}(\rho+r\Id)^{\alpha-1}\big]  -(1+r)^{\alpha-1}  \right]  \d r \\
        &= \frac{\alpha}{(1-\alpha)} \int_0^\infty  
        \tr\left[\sigma^{1-\alpha}\left[\big(\rho+r\Id\big)^{\alpha-1}-\big(\sigma+r\Id\big)^{\alpha-1}\right]\right]   \d r ,
    \end{align}
    where in the third line we use $H_\alpha(\sigma\|\sigma)=0$. We continue
    \begin{align}        
        H_\alpha(\rho\|\sigma) 
        &\geq  \frac{\alpha}{(1-\alpha)} \lim_{M\rightarrow\infty} \int_0^M  
        \tr\left[\sigma^{1-\alpha}\left[\big(\rho+r\Id\big)^{\alpha-1}-\big(\sigma+r\Id\big)^{\alpha-1}\right]\right]   \d r  \\
        &= \frac{1}{(1-\alpha)} \lim_{M\rightarrow\infty}  
        \tr\Big[\sigma^{1-\alpha}\left[\big(\rho+M\Id\big)^{\alpha}-\rho^\alpha- \big(\sigma+M\Id\big)^{\alpha} +\sigma^\alpha\right]\Big] 
        \\
        &= \frac{1}{(1-\alpha)} \left(
        1-\tr[\rho^\alpha\sigma^{1-\alpha} ] + \lim_{M\rightarrow\infty}  \tr\left[\sigma^{1-\alpha}\left[\big(\rho+M\Id\big)^{\alpha}- \big(\sigma+M\Id\big)^{\alpha} \right]\right]   \right)  \\
        &= \frac{1}{(1-\alpha)} \left(
        1-\tr[\rho^\alpha\sigma^{1-\alpha} ]   \right) \\
        &= \widebar H_\alpha(\rho\|\sigma) ,
    \end{align}
    where the penultimate line follows using $|(a+M)^\alpha - (b+M)^\alpha| = O(M^{\alpha-1})$ and $\alpha<1$.
    This proves the first inequality and the second one follows from the definition of the \Renyi divergences. 
\end{proof}
This result implies a more direct proof of the bound, 
\begin{align}
    D^\reg_\alpha(\rho\|\sigma)=\lim_{n\to\infty} D_\alpha\big(\rho^{\otimes n}\big\| \sigma^{\otimes n}\big) \geq \widebar D_\alpha(\rho\|\sigma), 
\end{align}
previously proven in~\cite[Section 3.1.4]{hirche2023quantum}. The previous result was proven using the inequality by Audenaert et al. which necessitated handling additional remainder terms. We will see in the following section that the new \Renyi divergences can indeed be used to strengthen the inequality by Audenaert et al. itself. 

Before that, we will give a similar result relating the \Renyi divergence with the sandwiched \Renyi divergence~\cite{muller2013quantum,wilde2014strong} for all integer $\alpha\geq2$. 
\begin{proposition}\label{Prop:Ha-sandHa}
    For any integer $\alpha\geq2$ we have
    \begin{align}
        H_\alpha(\rho\|\sigma) \leq \widetilde H_\alpha(\rho\|\sigma), \\ 
        D_\alpha(\rho\|\sigma) \leq \widetilde D_\alpha(\rho\|\sigma),
    \end{align}
        where $\widetilde H_\alpha(\rho\|\sigma) = \frac{1}{1-\alpha}\big(1-\tr\big[ (\sigma^{\frac{1-\alpha}{2\alpha}}\rho\sigma^{\frac{1-\alpha}{2\alpha}})^\alpha\big]\big)$ and $\widetilde D_\alpha(\rho\|\sigma) = \frac{1}{\alpha-1}\log \tr\big[ (\sigma^{\frac{1-\alpha}{2\alpha}}\rho\sigma^{\frac{1-\alpha}{2\alpha}})^\alpha\big]$ is the sandwiched \Renyi divergence.
\end{proposition}
\begin{proof}
We once again use the Araki-Lieb-Thirring inequality~\cite{lieb76,araki1990inequality} to write
\begin{align}
    &\int_0^\infty \tr\left( \rho(\sigma+s\Id)^{-1} \right)^\alpha \d s \\
   &= \int_0^\infty \tr\left( (\sigma+s\Id)^{-\frac{1}{2}} \sigma^{-\frac{1-\alpha}{2\alpha}}  \Big(\sigma^{\frac{1-\alpha}{2\alpha}} \rho 
  \sigma^{\frac{1-\alpha}{2\alpha}}\Big) \sigma^{-\frac{1-\alpha}{2\alpha}}(\sigma+s\Id)^{-\frac{1}{2}} \right)^\alpha \d s \\
  &\leq \int_0^\infty \tr\left( (\sigma+s\Id)^{-\alpha} \sigma^{-(1-\alpha)}  \Big(\sigma^{\frac{1-\alpha}{2\alpha}} \rho 
  \sigma^{\frac{1-\alpha}{2\alpha}}\Big)^\alpha  \right) \d s \\  
  & = \frac{1}{\alpha-1}\tr\left( \sigma^{1-\alpha} \sigma^{-(1-\alpha)}  \Big(\sigma^{\frac{1-\alpha}{2\alpha}} \rho 
  \sigma^{\frac{1-\alpha}{2\alpha}}\Big)^\alpha  \right)\\
   &= \frac{1}{\alpha-1} \tr  \Big(\sigma^{\frac{1-\alpha}{2\alpha}} \rho 
  \sigma^{\frac{1-\alpha}{2\alpha}}\Big)^\alpha  .
\end{align}
Using this in Theorem~\ref{Thm:Ha-2int} the desired inequalities follow. 
\end{proof}

We conjecture that Proposition~\ref{Prop:Ha-sandHa} does indeed hold for all $\alpha>1$.

\subsection{Quantum Chernoff bound}

In the remainder of this section, we consider the problem of symmetric quantum state discrimination. In particular, the optimal asymptotic error exponent for this problem, which is given by the quantum Chernoff bound~\cite{audenaert2007discriminating,nussbaum2009chernoff}. 
The core ingredient in proving the quantum Chernoff bound is the following inequality proven by Audenaert et al. in~\cite{audenaert2007discriminating} (see also~\cite{audenaert2008asymptotic}). For all $0\leq \alpha\leq 1$ and positive semidefinite operators $A$ and $B$ one has
\begin{align}
   \frac{1}{2} \tr\big[A+B-|A-B|\big] \leq \tr[A^\alpha B^{1-\alpha}]. \label{Eq:Audenaert}
\end{align}
In the special case of $A=\rho$ and $B=\sigma$, the inequality by Audenaert et al. turns into   
\begin{align}
    (1-\alpha) \widebar\Hell_\alpha(\rho\|\sigma) \leq E_1(\rho\|\sigma),
\end{align}
where as before $\widebar H_\alpha(\rho\|\sigma) = \frac{1}{1-\alpha}\big(1-\tr\big[\rho^\alpha\sigma^{1-\alpha}\big]\big)$. 
In the following lemma, we prove a stronger form of this inequality. 

\begin{lemma}\label{Lem:Aud-stronger}
For $0\leq \alpha\leq 1$, 
    \begin{align}
        (1-\alpha) \widebar\Hell_\alpha(\rho\|\sigma)\leq (1-\alpha) \Hell_\alpha(\rho\|\sigma) \leq E_1(\rho\|\sigma). 
    \end{align}
\end{lemma}
\begin{proof}
    The first inequality is established in Proposition~\ref{Prop:Ha-geq-tHa}. We prove the second one.
    \begin{align}
        (1-\alpha) \Hell_\alpha(\rho\|\sigma) &= \alpha(1-\alpha) \int_1^\infty \Big(\gamma^{\alpha-2} E_\gamma(\rho\|\sigma) + \gamma^{-\alpha-1} E_\gamma(\sigma\|\rho)\Big) \,\d\gamma \\
        &\leq \alpha(1-\alpha) E_1(\rho\|\sigma) \int_1^\infty\Big( \gamma^{\alpha-2}  + \gamma^{-\alpha-1}\Big) \,\d\gamma \\
        &= \alpha(1-\alpha) E_1(\rho\|\sigma) \left( \frac1{1-\alpha} + \frac{1}{\alpha}\right) \\
        &= E_1(\rho\|\sigma). 
    \end{align}
\end{proof}

Recall that $p_e(\rho,\sigma) = 1-\frac{1}{2} \tr|\rho-\sigma| $ is the average error in discriminating the states $\rho, \sigma$. Then the above lemma implies
\begin{align}
    p_e(\rho,\sigma) \leq 1-(1-\alpha) H_\alpha(\rho\|\sigma)\leq 1-(1-\alpha) \widebar H_\alpha(\rho\|\sigma) = \tr[\rho^\alpha\sigma^{1-\alpha}],
\end{align}
and 
\begin{align}
       -\frac 1n\log p_e\big(\rho^{\otimes n}, \sigma^{\otimes n}\big)\leq -\min_{0\leq\alpha\leq1}\log\tr[\rho^\alpha\sigma^{1-\alpha}] =: C(\rho,\sigma),
\end{align}
which is the achievability part of the quantum Chernoff bound. Here $C(\rho,\sigma)$ is the quantum Chernoff information divergence.

\section{Monotonicity and convexity properties}\label{Sec:Monotonicity}

In this section we discuss additional properties of $f$-divergences, particularly those that allow relating $f$-divergences for different functions. 

\subsection{Ordering of $f$-divergences}

We start with the following alternative expression for the $f$-divergence~\cite{hirche2023quantum},
\begin{align}
    D_f(\rho\|\sigma) = \int_0^\infty f''(\gamma) \tilde E_\gamma(\rho\|\sigma) \d\gamma, \label{Eq:f-div-alt}
\end{align}
where
\begin{align}
    \tilde E_\gamma(\rho\|\sigma) = E_\gamma(\rho\|\sigma) -(1-\gamma)_+. 
\end{align}
The main result of this section is the following theorem. 
\begin{theorem}\label{Thm:f-g-div-ineq}
    For any two twice differentiable functions $f,g$, such that 
    \begin{align}
        f(x)\geq g(x) \qquad\forall x\in\big[e^{-D_{\max}(\sigma\|\rho)},e^{D_{\max}(\rho\|\sigma)}\big],
    \end{align} 
    we have
    \begin{align}
        D_f(\rho\|\sigma) + f(1) \geq D_g(\rho\|\sigma) + g(1).
    \end{align}
\end{theorem}
\begin{proof}
    First, note that we can restrict the integral in Equation~\eqref{Eq:f-div-alt} to go from $\beta_2=e^{-D_{\max}(\sigma\|\rho)}$ up to $\beta_1^{-1}=e^{D_{\max}(\rho\|\sigma)}$. This is because the Hockey-Stick divergence $E_\gamma(\rho\|\sigma)$ vanishes for $\gamma\geq e^{D_{\max}(\rho\|\sigma)}$. 
    
    Let 
    \begin{align} p(x,t)=\det(t\Id-(\rho-x\sigma)),
    \end{align}
    be the characteristic polynomial of $\rho-x\sigma$, and let $\partial_t p(x,t)$ be its derivative with respect to $t$. Considering $p(x,t)$ and $\partial_t p(x,t)$ as polynomials of $t$, their resultant is a polynomial of $x$ with finitely many roots. This means that there are finitely many values of $x\in\RR$ for which $p(x,t)$ and $\partial_t p(x,t)$ have a common root. 
    Equivalently, except for finitely many values of $x\in\RR$, all the roots of $p(x,t)$, as a polynomial of $t$, are simple. Since $p(x,t)$ is the characteristic polynomial of $\rho-x\sigma$, we conclude that all the eigenvalues of $\rho-x\sigma$ are simple except for finitely many values of $x$. 

    Let $(a,b)\subseteq [\beta_2,\beta_1^{-1}]$ be an interval such that all eigenvalues of $\rho-x\sigma$ are simple for any $x\in(a,b)$. Then, by the analytic implicit function theorem, there are analytic functions 
    \begin{align}
        \lambda_1(x)>\lambda_2(x)>\dots>\lambda_d(x),
    \end{align}
    on $(a,b)$ that represent all the eigenvalues of $\rho-x\sigma$ for any $x\in(a,b)$, where $d$ is the dimension of the underlying Hilbert space. 

    Since $x\mapsto\rho-x\sigma$ is non-increasing in the L\"owner order, $\lambda_j(x)$ is non-increasing for any $j$. Therefore, for any $j$, one of the following holds: 
    \begin{itemize}
    \item[(i)] $\lambda_j(x)>0$ for all $x\in(a,b)$, \item[(ii)] $\lambda_j(x)<0$ for all  $x\in(a,b)$ 
    \item[(iii)] there is a unique $c_j\in(a,b)$ such that $\lambda_j(x)>0$ for any $a<x\leq c_j$ and $\lambda_j(x)\leq 0$ for any $c_j<x<b$. 
    \end{itemize}

    Recall that $E_\gamma(\rho\|\sigma)$ is equal to the summation of positive eigenvalues of $\rho-\gamma\sigma$. Thus, by the above arguments, the whole interval $[\beta_2,\beta_1^{-1}]$ can be divided into finitely many subintervals with endpoints
    \begin{align}
        \beta_2=a_0 < a_1 < \dots < a_{N+1}= \beta_1^{-1}
    \end{align}
    such that for any $0\leq i\leq N$ there is $0\leq k_i\leq d$ satisfying 
    \begin{align}
        E_\gamma(\rho\|\sigma) = \sum_{j=1}^{k_i} \lambda_j(\gamma)\qquad \forall x\in(a_i,a_{i+1}).
    \end{align}
    Moreover, $E_\gamma(\rho\|\sigma)$, and hence $\tilde E_\gamma(\rho\|\sigma)$,are analytical on any subinterval $(a_i,a_{i+1})$. Without loss of generality we also assume that $1\in\{a_1,\dots,a_N\}$. 

    With the above findings we now apply integration by parts to derive an equivalent formula for $D_f(\rho\|\sigma)$. 
    \begin{align}
        D_f(\rho\|\sigma) &= \int_{\beta_2}^{\beta_1^{-1}} f''(\gamma) \tilde E_\gamma(\rho\|\sigma) \d\gamma \\
        &= \sum_{i=0}^N \int_{a_i}^{a_{i+1}} f''(\gamma) \tilde E_\gamma(\rho\|\sigma) \d\gamma \\
        &= \sum_{i=0}^N \left( f'(a_{i+1}) \tilde E_{a_{i+1}}(\rho\|\sigma) - f'(a_{i}) \tilde E_{a_{i}}(\rho\|\sigma) - \int_{a_i}^{a_{i+1}} f'(\gamma) \tilde E'_\gamma(\rho\|\sigma) \d\gamma \right) \\
        &= f'(\beta_1^{-1}) \tilde E_{\beta_1^{-1}}(\rho\|\sigma) - f'(\beta_2) \tilde E_{\beta_2}(\rho\|\sigma) - \sum_{i=0}^N \left(  \int_{a_i}^{a_{i+1}} f'(\gamma) \tilde E'_\gamma(\rho\|\sigma) \d\gamma \right) \\
        &=  - \sum_{i=0}^N \left(  \int_{a_i}^{a_{i+1}} f'(\gamma) \tilde E'_\gamma(\rho\|\sigma) \d\gamma \right), 
    \end{align}
    where $E'_\gamma(\rho\|\sigma)$ is the derivative of $E_\gamma(\rho\|\sigma)$with respect to $\gamma$. Here, the third equality follows by integration by parts, the forth line follows by evaluation of the telescopic sum and the fifth line holds since $\tilde E_{\beta_1^{-1}}(\rho\|\sigma)=\tilde E_{\beta_2}(\rho\|\sigma)=0$. From the definition of $\tilde E_\gamma(\rho\|\sigma)$, we have
    \begin{align}
        \tilde E'(\rho\|\sigma) = \begin{cases}
            E'_\gamma(\rho\|\sigma) &\gamma\geq1 \\
            E'_\gamma(\rho\|\sigma) + 1 & 0\leq\gamma< 1
        \end{cases}. 
    \end{align}
    Therefore, 
    \begin{align}
        D_f(\rho\|\sigma) &=  - \sum_{i=0}^N \left(  \int_{a_i}^{a_{i+1}} f'(\gamma) \tilde E'_\gamma(\rho\|\sigma) \d\gamma \right) \\
        &= -\int_0^1 f'(\gamma) \d\gamma - \sum_{i=0}^N \left(  \int_{a_i}^{a_{i+1}} f'(\gamma)  E'_\gamma(\rho\|\sigma) \d\gamma \right) \\
        &= f({\beta_2}) - f(1) - \sum_{i=0}^N \left(  \int_{a_i}^{a_{i+1}} f'(\gamma)  E'_\gamma(\rho\|\sigma) \d\gamma \right) \\
        &= f({\beta_2}) - f(1) - \sum_{i=0}^N \left(  f(a_{i+1}) 
        E'_{a_{i+1},-}(\rho\|\sigma) - f(a_{i}) 
        E'_{a_{i},+}(\rho\|\sigma) - \int_{a_i}^{a_{i+1}} f(\gamma)  E''_\gamma(\rho\|\sigma) \d\gamma \right), 
    \end{align}
    where in the last equality we again apply integration by parts, and $E'_{a_{i+1},-}(\rho\|\sigma)$ and $E'_{a_{i},+}(\rho\|\sigma)$ denote the corresponding left and right derivatives, respectively. We continue with 
    \begin{align}
        D_f(\rho\|\sigma) &= f({\beta_2}) - f(1) - \sum_{i=0}^N \left(  f(a_{i+1}) 
        E'_{a_{i+1},-}(\rho\|\sigma) - f(a_{i}) 
        E'_{a_{i},+}(\rho\|\sigma) - \int_{a_i}^{a_{i+1}} f(\gamma)  E''_\gamma(\rho\|\sigma) \d\gamma \right) \\
        &= - f(1) +f({\beta_2})\left(1+E'_{{\beta_2},+}(\rho\|\sigma)\right) - f(\beta_1^{-1}) E'_{\beta_1^{-1},-}(\rho\|\sigma) + \sum_{i=1}^N f(a_{i})\left(   
        E'_{a_{i},+}(\rho\|\sigma) - 
        E'_{a_{i},-}(\rho\|\sigma) \right) \nonumber\\
        &\qquad + \sum_{i=0}^N \left( \int_{a_i}^{a_{i+1}} f(\gamma)  E''_\gamma(\rho\|\sigma) \d\gamma \right) \label{Eq:f-div-another-rep}
    \end{align}
    where in the final equality we have simply reshuffled terms.  
    Finally, recall that $\gamma\mapsto E_\gamma(\rho\|\sigma)$ is convex, and hence $E''_\gamma(\rho\|\sigma)\geq 0$ for any $\gamma\in(a_i,a_{i+1})$. By the same argument, we also have 
    \begin{align}
        -1\leq E'_{a_{i},-}(\rho\|\sigma) \leq E'_{a_{i},+}(\rho\|\sigma) \leq 0 .
    \end{align}
    From the above representation of the $f$-divergence, the claimed result follows immediately. 
\end{proof}

The proven monotonicity in the function $f$ is of course natural for any good definition of $f$-divergence. However, its proof for this definition was an open problem. The key new tool in our proof is the representation in Equation~\eqref{Eq:f-div-another-rep}  which can be further comprehended by considering the classical case. For probability distributions $p(i)$ and $q(i)$ on a finite set, the Hockey-Stick divergence $E_\gamma(p\|q)$ is piecewise linear and the $a_i$'s are given by $a_i=\frac{p(i)}{q(i)}$. Hence, it is easily verified that for all $\gamma\in(a_i,a_{i+1})$ we have $E''_\gamma(\rho\|\sigma)=0$. 
Furthermore, the first derivative can easily be calculated, leading to
\begin{align}
    E'_{a_{i},+}(p\|q) - E'_{a_{i},-}(p\|q) = - \sum_{j\leq i-1} q(j) + \sum_{j\leq i} q(j) = q(i). 
\end{align}
We are only missing two terms which are given by 
\begin{align}
    f({\beta_2})\left(1+E'_{{\beta_2},+}(\rho\|\sigma)\right) &= f\left(e^{-D_{\max}(\sigma\|\rho)}\right)\left(1-\sum_{x\leq N} q(x)\right) = q(N+1) f\left(\frac{p(N+1)}{q(N+1)}\right),
\end{align}
and
\begin{align}
    - f(\beta_1^{-1}) E'_{\beta_1^{-1},-}(\rho\|\sigma) &= q(0) f\left(\frac{p(0)}{q(0)}\right).
\end{align}
In summary, for two probability distributions, the representation in Equation~\eqref{Eq:f-div-another-rep} evaluates to 
\begin{align}
    D_f(p\|q) &= - f(1) + \sum_{i=0}^{N+1} q(i) f\left(\frac{p(i)}{q(i)}\right), 
\end{align}
which is exactly the usual definition of the classical $f$-divergence. 

Next, the general result in Theorem~\ref{Thm:f-g-div-ineq} allows us to give a quantum generalization of~\cite[Theorem 1 \& Theorem 6]{sason2016f}. 

\begin{theorem}\label{Thm:f-g-bounds-opt}
    Fix $0<\beta_1, \beta_2<1$. Given functions $f, g\in\cF$, assume that $g(t)>0$ for all $t\in(\beta_2,1)\cup(1,\beta_1^{-1})$. Define the function $\kappa(t):\in(0,1)\cup(1,\infty)\mapsto\RR$ by
    \begin{align}
        \kappa(t)=\frac{f(t)}{g(t)},
    \end{align}
    and let
    \begin{align}         \bar\kappa &= \sup_{t\in(\beta_2,1)\cup(1,\beta_1^{-1})} \kappa(t).
    \end{align}
    Then, the followings hold:
    \begin{enumerate}
        \item[(i)] If  $\beta_1\leq e^{-D_{\max}(\rho\|\sigma)}$ and $\beta_2\leq e^{-D_{\max}(\sigma\|\rho)}$, then $D_f(\rho\|\sigma) \leq \bar\kappa D_g(\rho\|\sigma)$.
        \item[(ii)] If  $f'(1)=g'(1)=0$, then
        \begin{align}
            \sup_{\rho\neq\sigma} \frac{D_f(\rho\|\sigma)}{D_g(\rho\|\sigma)}=\bar\kappa, \label{Eq:f-g-bound-sup-opt}
        \end{align}
        where the supremum is taken over all quantum states $\rho, \sigma$ satisfying $\beta_1\leq e^{-D_{\max}(\rho\|\sigma)}$ and $\beta_2\leq e^{-D_{\max}(\sigma\|\rho)}$.
    \end{enumerate}
\end{theorem}
\begin{proof}
    The first claim follows directly from Theorem~\ref{Thm:f-g-div-ineq} and then also implies the inequality 
    in one direction in the second claim. To prove the inequality in the other direction, it is enough to construct a pair of states that achieves the value $\bar\kappa$. Since the $f$-divergence reduces to the usual classical divergence for commuting states, the result follows directly from~\cite[Theorem 1b]{sason2016f} (see also~\cite[Remark 11]{sason2016f}).
\end{proof}
This allows for the generalization of several inequalities between $f$-divergences originally proven in~\cite{sason2016f} for the classical setting. As a particularly relevant example, we state a quantum 
variant of~\cite[Theorem 9]{sason2016f}. 
\begin{corollary}
    For $\alpha\in(0,1)\cup(1,\infty)$ define the function $\kappa_\alpha(t)$ on $[0, +\infty]$ by
    \begin{align}
    \kappa_\alpha(t)=\begin{cases}
            1 &t=0\\
            \frac{(1-\alpha)r(t)}{1-t^\alpha+\alpha (t-1)} &t\in(0,1)\cup(1,\infty) \\
            \alpha^{-1} &t=1 \\
            \infty & t=\infty, \alpha\in(0,1) \\
            0 & t=\infty, \alpha\in(1,\infty)
        \end{cases},
    \end{align}
    where $r(t)=t\log(t)+(1-t)$.   Let $\rho\neq\sigma$ be such that $\rho\ll\sigma$ and let $\beta_1= e^{-D_{\max}(\rho\|\sigma)}$ and $\beta_2= e^{-D_{\max}(\sigma\|\rho)}$.  
    Then, for $\alpha\in(0,1)$, 
    \begin{align}
        \kappa_\alpha(\beta_2) \leq \frac{D(\rho\|\sigma)}{H_\alpha(\rho\|\sigma)} \leq \kappa_\alpha(\beta_1^{-1}), 
    \end{align}
    and for $\alpha\in(1,\infty)$, 
    \begin{align}
        \kappa_\alpha(\beta_1^{-1}) \leq \frac{D(\rho\|\sigma)}{H_\alpha(\rho\|\sigma)} \leq \kappa_\alpha(\beta_2). 
    \end{align}
\end{corollary}
\begin{proof}
    The proof relies solely on the properties of the functions defining the divergences, and is essentially identical to the proof of~\cite[Theorem 9]{sason2016f}. 
\end{proof}
Special cases of the above corollary include 
\begin{align}
    H_\frac12(\rho\|\sigma) \leq \kappa_{\frac12} (\beta_2) H_\frac12(\rho\|\sigma) &\leq D(\rho\|\sigma) \leq \kappa_{\frac12} (\beta_1^{-1}) H_\frac12(\rho\|\sigma), \\
    0\leq \kappa_2(\beta_1^{-1}) \chi^2(\rho\|\sigma) &\leq D(\rho\|\sigma) \leq \kappa_2(\beta_2) \chi^2(\rho\|\sigma) \leq \chi^2(\rho\|\sigma). 
\end{align}
Note that the bound $\kappa_2(\beta_1^{-1}) \chi^2(\rho\|\sigma) \leq D(\rho\|\sigma)$ was previously proven in~\cite[Lemma 2.2]{gao2022complete}. From Theorem~\ref{Thm:f-g-bounds-opt} we can add that this bound is optimal in the sense of Equation~\eqref{Eq:f-g-bound-sup-opt}. For further examples we refer to~\cite[Section IV]{sason2016f}. 

\subsection{Convexity and monotonicity of the Hellinger divergences}

In this section we prove additional results specifically for Hellinger and \Renyi divergences. 
\begin{proposition}
    $\alpha\mapsto \frac 1\alpha H_\alpha(\rho\|\sigma)$ is log-convex in $(0, +\infty)$. That is, for $\alpha,\beta>0$, 
    \begin{align}
        H_{\frac{\alpha+\beta}2}(\rho\|\sigma)^2 \leq \frac{(\alpha+\beta)^2}{4\alpha\beta} H_{\alpha}(\rho\|\sigma)\,H_{\beta}(\rho\|\sigma). 
    \end{align}
\end{proposition}
\begin{proof}
    From the definition of the quantum $f$-divergence in Equation~\eqref{Eq:f-divergence} and that $H_\alpha(\rho\|\sigma) = D_{f_\alpha}(\rho\|\sigma)$ with $f_\alpha(\gamma) = \frac{\gamma^\alpha-1}{\alpha-1}$ we realize that
    \begin{align}
        \frac 1\alpha H_\alpha(\rho\|\sigma) = \int_0^{+\infty} \varphi(\gamma) \gamma^\alpha \d \gamma ,
    \end{align}
    where $\varphi(\gamma)$ is a non-negative function depending on $\rho, \sigma$ but not on $\alpha$. 
    Let $0<\alpha_0< \alpha_1< \alpha_2$ be such that $\alpha_1 = \frac12(\alpha_0+\alpha_2)$. Then, by the Cauchy-Schwarz inequality we have 
    \begin{align}
        \frac 1{\alpha_1} H_{\alpha_1}(\rho\|\sigma) &= \int_0^{+\infty} \varphi(\gamma) \gamma^{\alpha_1} \d \gamma \\
        &\leq \left( \int_0^{+\infty} \varphi(\gamma) \gamma^{\alpha_0}  \right)^{1/2}\cdot \left( \int_0^{+\infty} \varphi(\gamma) \gamma^{\alpha_2}  \right)^{1/2}\\
        &= \left(\frac{1}{\alpha_0}H_{\alpha_0}(\rho\|\sigma)   \right)^{1/2} \cdot \left(\frac{1}{\alpha_2}H_{\alpha_2}(\rho\|\sigma)   \right)^{1/2}\,,
    \end{align}
that is equivalent to the desired log-convexity.
\end{proof}

The log-convexity of $\alpha\mapsto \frac 1\alpha H_\alpha(\rho\|\sigma)$ of course implies the convexity of $\alpha\mapsto \frac 1\alpha H_\alpha(\rho\|\sigma)$. This in particular gives
    \begin{align}
        D(\rho\|\sigma) \geq H_\alpha(\rho\|\sigma) - \frac{\alpha-1}{\alpha+1}H_{\alpha+1}(\rho\|\sigma), 
    \end{align}
for any $\alpha\geq1$. In particular, for $\alpha=2$ get
    \begin{align}\label{eq:chi2-H3-D-log-convexity}
        \chi^2(\rho\|\sigma) - \frac13 H_3(\rho\|\sigma) \leq D(\rho\|\sigma).
    \end{align}

The above result is a quantum generalization of~\cite[Theorem 36 c)]{sason2016f}. It was also observed in~\cite[Theorem 36 a)]{sason2016f} that the classical Hellinger divergence is monotonically increasing in $\alpha$. We generalize this to the quantum setting in the following. 
\begin{corollary}
    The following properties hold, 
    \begin{enumerate}
        \item The function $\alpha\mapsto H_\alpha(\rho\|\sigma)$ is monotonically increasing on $(0,\infty)$, 
        \item The function $\alpha\mapsto \left(\frac1\alpha-1\right) H_\alpha(\rho\|\sigma)$ is monotonically decreasing on $(0,1)$. 
    \end{enumerate}
\end{corollary}
\begin{proof}
    The proofs follow analogously to those of 
    the classical setting in~\cite[Theorem 36]{sason2016f}, using Theorem~\ref{Thm:f-g-div-ineq}. 
\end{proof}
In the next step we want to also prove monotonicity for \Renyi divergences.
Recall that
\begin{align}
    D_\alpha(\rho\|\sigma) = \frac1{\alpha-1} \log Q_\alpha(\rho\|\sigma) = \frac1{\alpha-1} \log\left( 1 + (\alpha-1)H_\alpha(\rho\|\sigma)\right). 
\end{align}

\begin{theorem}
    The following properties hold, 
    \begin{enumerate}
        \item The function $\alpha\mapsto Q_\alpha(\rho\|\sigma)$ is log-convex for $\alpha\in(0,\infty)$, 
        \item The function $\alpha\mapsto D_\alpha(\rho\|\sigma)$ is monotonically increasing in $\alpha\in(0,\infty)$. 
    \end{enumerate}
\end{theorem}
\begin{proof}
    Recall that $H_\alpha(\rho\|\sigma)$ is the $f$-divergence for $f(x)=\frac{x^\alpha-1}{\alpha-1}$. Using this in Equation~\eqref{Eq:f-div-another-rep}, we get
    \begin{align}
        Q_\alpha(\rho\|\sigma) &= 1 + \sum_{i=1}^N (a_i^\alpha-1) \left(   
        E'_{a_{i},+}(\rho\|\sigma) - 
        E'_{a_{i},-}(\rho\|\sigma) \right) + \sum_{i=0}^N \left( \int_{a_i}^{a_{i+1}} (\gamma^\alpha-1)  E''_\gamma(\rho\|\sigma) \d\gamma \right) \\
        &= \sum_{i=1}^N a_i^\alpha \left(   
        E'_{a_{i},+}(\rho\|\sigma) - 
        E'_{a_{i},-}(\rho\|\sigma) \right) + \sum_{i=0}^N \left( \int_{a_i}^{a_{i+1}} \gamma^\alpha  E''_\gamma(\rho\|\sigma) \d\gamma \right), 
    \end{align}
    which holds because the removed terms correspond to the $f$-divergence for $f(x)=-1$, which is trivial. We conclude that there is a measure $\d\varphi(\gamma)$, such that, 
    \begin{align}
        Q_\alpha(\rho\|\sigma) &= \|\id\|_\alpha^\alpha, 
    \end{align}
    where $\id(x)=x$ is the identity function and 
    \begin{align}
        \|g\|_\alpha = \left( \int |g(x)|^\alpha \d\varphi(x) \right)^\frac1\alpha .
    \end{align}
    Let $\alpha_0,\alpha_1>0$ and $\theta\in(0,1)$ be such that $\alpha=\theta \alpha_0+(1-\theta)\alpha_1$. Then by Lyapunov's inequality, we have
    \begin{align}
        \|\id\|_\alpha^\alpha \leq \|\id\|_{\alpha_0}^{\theta\alpha_0}\,\|\id\|_{\alpha_1}^{(1-\theta)\alpha_1}, 
    \end{align}
    which shows log-convexity of $Q_\alpha(\rho\|\sigma)$ and hence the first claim. The log-convexity then directly implies the monotonicity in the second claim by a standard argument (see, e.g.,~\cite[Corollary 4.3]{tomamichel2015quantum}). 
\end{proof}
As a special case, this implies a strengthening of a previously discussed inequality, 
\begin{align}
    D(\rho\|\sigma) \leq D_2(\rho\|\sigma) = \log(1+\chi^2(\rho\|\sigma)) \leq \chi^2(\rho\|\sigma). 
\end{align}

\section{Further inequalities from alternative integral representations}\label{Sec:More-stuff}

In this final section, we will give some alternative integral representations that connect $f$-divergences for different functions. We will use these representations to prove additional inequalities between $f$-divergences. 

\subsection{Integral representations in terms of Hellinger divergence}

In the following we will discuss an integral of a divergence that always produces again a valid divergence. Later, we will see that this has interesting implications when considering Hellinger divergences. This result is motivated by~\cite[Theorem 1 \& 4]{nishiyama2020relations}. However, the results are much more general and mostly even new in the classical setting. 

\begin{proposition}\label{prop:integral-Df-DF}
    We have for any $f\in\cF$ and $\lambda\in[0,1]$, 
    \begin{align}
        D_F(\rho\|\sigma) = \int_0^\lambda D_f\big(\rho\big\|(1-s)\rho+s\sigma\big) \frac{\d s}{s},
    \end{align}
    where $F(\gamma)$ is such that
    \begin{align}
        F''(\gamma) =\frac{(t_\lambda-1) f'(t_\lambda)-f(t_\lambda)}{\gamma(\gamma-1)^2}, 
    \end{align}
    with $t_\lambda=\frac{\gamma}{\lambda+(1-\lambda)\gamma}$. In particular, $F(\gamma)$ is always a convex function. 
\end{proposition}
\begin{proof}
    As shown in~\cite{hirche2023quantum} we have 
    \begin{align}
        D_f\big(\rho\big\|(1-s)\rho+s\sigma\big) = D_{G_s}(\rho\|\sigma)
    \end{align}
    with $G_s(\gamma)=(s+(1-s)\gamma)f\big(\frac{\gamma}{s+(1-s)\gamma}\big)$. Taking the second derivative, we obtain 
    \begin{align}           G_s''(\gamma)=\frac{s^2}{(s+(1-s)\gamma)^3}f''\big(\frac{\gamma}{s+(1-s)\gamma}\big),\label{Eq:G-rsf}
    \end{align}
    which we need to integrate to get 
    \begin{align}
        \int_0^\lambda G_s''(\gamma) \frac{\d s}{s} =\int_0^\lambda \frac{s}{(s+(1-s)\gamma)^3} f''\big(\frac{\gamma}{s+(1-s)\gamma}\big) \d s. 
    \end{align}
     Substituting $t=\frac{\gamma}{s+(1-s)\gamma}$ leads to
    \begin{align}
        \int_0^\lambda G_s''(\gamma) \frac{\d s}{s} &=\int_1^{t_\lambda} \frac{(t-1)}{\gamma(\gamma-1)^2} f''(t)\d t
        =\frac{(t_\lambda-1) f'(t_\lambda)-f(t_\lambda)}{\gamma(\gamma-1)^2}= F''(\gamma), 
    \end{align}
    where $t_\lambda=\frac{\gamma}{\lambda+(1-\lambda)\gamma}$ and the second equality follows by integration by parts. Putting these together and using the definition of $D_F(\rho\|\sigma)$, we arrive at the desired equality. 
    
    To prove the convexity of $F(\gamma)$, we need to show that $F''(\gamma)\geq 0$. Recall that for differentiable $f$, the convexity of $f$ is equivalent to
    \begin{align}
        f(x) \geq f(y) + f'(y)(x-y).
    \end{align}
    Choosing $y=t_\lambda$ and $x=1$, we find that
    \begin{align}
     F''(\gamma)=\frac{f'(t_\lambda)(t_\lambda-1)-f(t_\lambda)}{\gamma(\gamma-1)^2}\geq 0,
    \end{align}
    for $\gamma>0$.
\end{proof}

It seems curious that the result is always a valid divergence, but we have to look at examples to see why this is potentially interesting. 

\begin{proposition}
    For any $\lambda\in[0,1]$ and integer $k\geq 2$ we have
    \begin{align}
        \sum_{\alpha=1}^{k-1}H_\alpha\big(\rho\big\|(1-\lambda)\rho+\lambda\sigma\big) 
        = (k-1) \int_0^\lambda H_k\big(\rho\big\|(1-s)\rho+s\sigma\big) \frac{\d s}{s} ,
    \end{align}
    with $H_1(\rho\|\sigma) = D(\rho\|\sigma)$. Hence,
    \begin{align}\label{eq:sum-H-alpha-int-H-k}
        \sum_{\alpha=1}^{k-1}H_\alpha(\rho\|\sigma)
        = (k-1) \int_0^1 H_k\big(\rho\big\|(1-s)\rho+s\sigma\big) \frac{\d s}{s} .
    \end{align}
\end{proposition}
\begin{proof}
Let $f(\gamma)=\frac{\gamma^k-1}{k-1}$ and use Proposition~\ref{prop:integral-Df-DF} to obttain
\begin{align}
(k-1) \int_0^\lambda H_k\big(\rho\big\|(1-s)\rho+s\sigma\big) \frac{\d s}{s} = (k-1)D_F(\rho\|\sigma),
    \end{align}
where
\begin{align}
    F''(\gamma) = \frac{(k-1)t_\lambda^{k}-kt_\lambda^{k-1}+1}{(k-1)(\gamma-1)^2\gamma}. 
\end{align}
While this does not seem very insightful, we can approach it from the other end. Note that
\begin{align}
    \sum_{\alpha=1}^{k-1} \alpha x^{\alpha-2} = \frac{(k-1)x^k-kx^{k-1}+1}{x(x-1)^2}, \label{Eq:Cool-Sum}
\end{align}
and hence 
\begin{align}
    \sum_{\alpha=1}^{k-1} \frac{\alpha \lambda^2 t_\lambda^{\alpha-2}}{(\lambda+(1-\lambda)\gamma)^3} 
    &= \frac{\lambda^2}{(\lambda+(1-\lambda)\gamma)^3}\frac{(k-1)t_\lambda^{k}-kt_\lambda^{k-1}+1}{(t_\lambda-1)^2t_\lambda} 
    = \frac{(k-1)t_\lambda^{k}-kt_\lambda^{k-1}+1}{(\gamma-1)^2\gamma},
\end{align}
where for the second equality we mostly need to check that $(t_\lambda-1)(\lambda+(1-\lambda)\gamma)=\lambda(\gamma-1)$. On the other hand, according to Equation~\eqref{Eq:G-rsf} we have $H_\alpha(\rho\| (1-\lambda)\rho+\lambda \sigma) = D_{G_{\alpha, \lambda}}(\rho\|\alpha)$ where
\begin{align}
    G''_{\alpha, \lambda} (\gamma) = \frac{\alpha\lambda^2 t_\lambda^{\alpha-2}}{(\lambda+(1-\lambda)\gamma)^3},
\end{align}
and the above calculations verify that $\sum_{\alpha=1}^{k-1} G''_{\alpha, \lambda} (\gamma) = (k-1)F''(\gamma)$. This gives the desired identity.

\end{proof}

This proposition in the special case of $k=2$ yields the following. 
\begin{corollary}\label{cor:D-int-chi2-lambda}
    We have for any  $\lambda\in[0,1]$ 
    \begin{align}
        D\big(\rho\big\|(1-\lambda)\rho+\lambda\sigma\big) = \int_0^\lambda \chi^2\big(\rho\big\|(1-s)\rho+s\sigma\big) \frac{\d s}{s},
    \end{align}
    and hence
    \begin{align}\label{eq:D-chi2-integral}
        D(\rho\|\sigma) = \int_0^1 \chi^2\big(\rho\big\|(1-s)\rho+s\sigma\big) \frac{\d s}{s} .
    \end{align}
\end{corollary}

This corollary gives a quantum version of~\cite[Theorem 1]{nishiyama2020relations}, which motivated this investigation. 
In~\cite{gao2022complete} also the following integral representation of the relative entropy was given, (in our notation) 
\begin{align}
    D(\rho\|\sigma) = \int_0^1 \int_0^s \frac{1}{(1-t)^2}\chi^2\big(\rho\big\|t\rho+(1-t)\sigma\big) \,\d t \,\d s .
\end{align}
With the above observation we get a somewhat simpler integral representation that is nevertheless in a similar spirit. 

Furthermore,~\cite[Theorem 1]{nishiyama2020relations} gives the following second integral formula for probability distributions $P$ and $Q$, 
\begin{align}
    \frac12 \lambda^2 \chi^2\big(\lambda P + (1-\lambda)Q \big\| Q\big) = \int_0^\lambda \chi^2\big(s P + (1-s)Q \big\| Q\big) \frac{\d s}{s}, 
\end{align}
which similar to the previous discussion, can be generalized as follows.

\begin{proposition}\label{prop:DF-Df-int-2nd}
    For $f\in\cF$ and  $\lambda\in[0,1]$, we have
    \begin{align}
        D_F(\rho\|\sigma) = \int_0^\lambda D_f\big(s\rho+(1-s)\sigma\big\|\sigma\big) \frac{\d s}{s},
    \end{align}
    where $F(\gamma)$ is such that,
    \begin{align}
        F''(\gamma) =\frac{(t_\lambda-1) f'(t_\lambda)-f(t_\lambda)}{(\gamma-1)^2}, 
    \end{align}
    with $t_\lambda=\lambda\gamma+(1-\lambda)$. In particular, $F(\gamma)$ is a convex function. 
\end{proposition}
\begin{proof}
    From~\cite{hirche2023quantum}, we can extract that 
    \begin{align}
        D_f(s\rho+(1-s)\sigma\|\sigma) = D_G(\rho\|\sigma), 
    \end{align}
    with $G(\gamma)=f(s\gamma+1-s)$ satisfying 
    \begin{align}
        G''(\gamma)=s^2 f''(s\gamma+1-s). 
    \end{align}
    Next, we verify that 
    \begin{align}
        \int_0^\lambda G''(\gamma) \frac{\d s}{s}=\int_0^\lambda s f''(s\gamma+1-s) \d s = \frac{\lambda(\gamma-1) f'(t_\lambda)-f(t_\lambda)}{(\gamma-1)^2}.
    \end{align}
    Noting that $\lambda(\gamma-1)=t_\lambda-1$, we find that $\int_0^\lambda G''(\gamma) \frac{\d s}{s} = F''(\gamma)$ and the desired equality holds.
    The convexity of $F(\gamma)$ follows by the same argument as in the proof of Proposition~\ref{prop:integral-Df-DF}.  
\end{proof}
Next, we specialize this again to Hellinger divergences. 

\begin{proposition}
    For any integer $k\geq 2$ and $\lambda\in[0,1]$, we have
    \begin{align}
         \sum_{\alpha=2}^{k} \frac{(\alpha-1)\lambda^2}{\alpha} H_\alpha\big(\lambda\rho+(1-\lambda)\sigma\big\|\sigma\big) = (k-1) \int_0^\lambda H_k\big(s\rho+(1-s)\sigma\big\|\sigma\big) \frac{\d s}{s}
    \end{align}
    and hence
    \begin{align}
         \sum_{\alpha=2}^{k} \frac{(\alpha-1)}{\alpha} H_\alpha(\rho\|\sigma) = (k-1) \int_0^1 H_k\big(s\rho+(1-s)\sigma\big\|\sigma\big) \frac{\d s}{s} .
    \end{align}
\end{proposition}
\begin{proof}
    Apply Proposition~\ref{prop:DF-Df-int-2nd} to $f_k(\gamma)=\frac{\gamma^k-1}{k-1}$ to get 
    \begin{align}
           F''(\gamma) = \frac{(k-1)t_\lambda^{k}-kt_\lambda^{k-1}+1}{(k-1)(\gamma-1)^2}
           = \lambda^2\frac{(k-1)t_\lambda^{k}-kt_\lambda^{k-1}+1}{(k-1)(t_\lambda-1)^2}.
    \end{align}
    By invoking the sum formula in Equation~\eqref{Eq:Cool-Sum}, we get 
     \begin{align}
           F''(\gamma) 
           = \frac{\lambda^2} {k-1}\sum_{\alpha=1}^{k-1} \alpha t_\lambda^{\alpha-1} 
           = \frac{\lambda^2}{k-1}\sum_{\alpha=2}^{k} (\alpha-1) t_\lambda^{\alpha-2} 
           = \frac{\lambda^2}{k-1}\sum_{\alpha=2}^{k} \frac{(\alpha-1)}{\alpha} f''_\alpha(t_\lambda).  
    \end{align}
    This implies the desired result. 
\end{proof}

For the special case of $k=2$ this proposition implies 
\begin{align}
    \frac12 \lambda^2 \chi^2\big(\lambda \rho + (1-\lambda)\sigma \big\| \sigma\big) = \int_0^\lambda \chi^2\big(s \rho + (1-s)\sigma \big\| \sigma\big) \frac{\d s}{s}, 
\end{align} 
which is the quantum generalization of the classical integral formula stated in~\cite[Theorem 1]{nishiyama2020relations}.


\subsection{Some resulting inequalities}

In this subsection, before exploring further integral representations, we establish some inequalities as consequences of the above integral representations. 

First, we note that the integrals in the previous results can always be upper bounded using the  convexity of divergences. For instance, $H_k(\rho\|(1-s)\rho+s\sigma)\leq (1-s)H_k(\rho\|\rho)+sH_k(\rho\|\sigma) = sH_k(\rho\| \sigma)$ and then
\begin{align}
    \int_0^\lambda H_k(\rho\|(1-s)\rho+s\sigma) \frac{\d s}{s} 
    \leq 
     \lambda H_k(\rho\|\sigma). 
\end{align}
Applying this to Equation~\eqref{eq:D-chi2-integral}, this implies the well known inequality
\begin{align}\label{eq:D-chi^2-ineq}
    D(\rho\|\sigma) \leq \chi^2(\rho\|\sigma). 
\end{align}

For some further simple applications, we prove the following corollary. For that, we recall the definition of the Jensen-Shannon divergence,
\begin{align}
JS(\rho\|\sigma) := \frac12 D\left(\rho\middle\|\frac{\rho+\sigma}{2}\right) + \frac12 D\left(\sigma\middle\|\frac{\rho+\sigma}{2}\right) .
\end{align}
\begin{corollary}
    We have 
    \begin{align}
        D\big(\rho\big\|(1-\lambda)\rho+\lambda\sigma\big) \leq \frac{\lambda^2}2 \left(\chi^2(\rho\|\sigma) + \chi^2(\sigma\|\rho)\right). 
    \end{align}
    and hence 
    \begin{align}
        D(\rho\|\sigma) &\leq \frac{1}2 \left(\chi^2(\rho\|\sigma) + \chi^2(\sigma\|\rho)\right) , \\
        JS(\rho\|\sigma) &\leq \frac{1}8 \left(\chi^2(\rho\|\sigma) + \chi^2(\sigma\|\rho)\right).
    \end{align}
\end{corollary}
\begin{proof}
    Invoking Corollary~\ref{cor:D-int-chi2-lambda} we have 
    \begin{align}
        D\big(\rho\big\|(1-\lambda)\rho+\lambda\sigma\big) &= \int_0^\lambda \chi^2\big(\rho\big\|(1-s)\rho+s\sigma\big) \frac{\d s}{s} \\        
        &= \int_0^\lambda \left(\chi^2(\rho\|(1-s)\rho+s\sigma) + \frac{s}{1-s}\chi^2(\sigma\|(1-s)\rho+s\sigma)\right) \d s\\
        &\leq \int_0^\lambda \left( s \chi^2(\rho\|\sigma) + s \chi^2(\sigma\|\rho)\right) \d s \\
        &= \frac{\lambda^2}2 \left(\chi^2(\rho\|\sigma) + \chi^2(\sigma\|\rho)\right),
    \end{align}
    where in the second line we invoke Lemma~\ref{lem:LC-X2-rel} and the inequality is due to convexity. 
\end{proof}

Taking a different route we get the following bound. 
\begin{corollary}
    We have 
    \begin{align}
        D\big(\rho\big\|(1-\lambda)\rho+\lambda\sigma\big) \leq \frac{\lambda^2}2 \chi^2(\rho\|\sigma) + \Big(\lambda+\frac{\lambda^2}{2}\Big)E_1(\sigma\|\rho), 
    \end{align}
    and hence
    \begin{align}
        D(\rho\|\sigma) &\leq \frac{1}2 \chi^2(\rho\|\sigma) + \frac32E_1(\rho\|\sigma) , \\
        JS(\rho\|\sigma) &\leq \frac{1}{16} \left(\chi^2(\rho\|\sigma) + \chi^2(\sigma\|\rho)\right) + \frac58 E_1(\rho\|\sigma). 
    \end{align}
\end{corollary}
\begin{proof}
    Once again we apply Corollary~\ref{cor:D-int-chi2-lambda} to get 
    \begin{align}
        D\big(\rho\big\|(1-\lambda)\rho+\lambda\sigma\big) &= \int_0^\lambda \chi^2\big(\rho\big\|(1-s)\rho+s\sigma\big) \frac{\d s}{s} \\        
        &= \int_0^\lambda s\chi^2(\rho\|\sigma) \d s + \int_0^\lambda \left(\chi^2\big(\rho\big\|(1-s)\rho+s\sigma\big) - s^2\chi^2(\rho\|\sigma) \right) \frac{\d s}{s}\\
        & =\frac{\lambda^2}2 \chi^2(\rho\|\sigma) +\int_0^\lambda \left(\chi^2\big(\rho\big\|(1-s)\rho+s\sigma\big) - s^2\chi^2(\rho\|\sigma) \right) \frac{\d s}{s} .
    \end{align}
    Next, for the second term, we use the shorthand notation $\bar s:=1-s$ to write
    \begin{align}
        &\int_0^\lambda \left(\chi^2\big(\rho\big\|(1-s)\rho+s\sigma\big) - s^2\chi^2(\rho\|\sigma) \right) \frac{\d s}{s} \\
        &= \int_0^\lambda 2 \int_1^\infty \Big(E_\gamma\big(\rho\big\|\bar s \rho+s\sigma\big) - s^2E_\gamma(\rho\|\sigma) + \gamma^{-3} E_\gamma\big(\bar s \rho+s\sigma\big\|\rho\big) - \gamma^{-3}s^2E_\gamma(\sigma\|\rho) \Big) \d\gamma \frac{\d s}{s} \\
        &= \int_0^\lambda 2 \left[\int_1^{\frac1{1-s}} \Big(E_\gamma\big(\rho\big\|\bar s\rho+s\sigma\big) - s^2E_\gamma(\rho\|\sigma)\Big)\d\gamma +  \int_1^\infty \gamma^{-3} \Big(E_\gamma\big(\bar s\rho+s\sigma\big\|\rho\big) - s^2E_\gamma(\sigma\|\rho) \Big) \d\gamma \right]\frac{\d s}{s} \\
        &\leq \int_0^\lambda 2 \left[\int_1^{\frac1{1-s}} \left(s E_\gamma(\rho\|\sigma) - s^2E_\gamma(\rho\|\sigma)\right)\d\gamma +  \int_1^\infty \gamma^{-3} \left( s E_\gamma(\sigma\|\rho) - s^2E_\gamma(\sigma\|\rho) \right) \d\gamma \right]\frac{\d s}{s} \\
        &\leq \int_0^\lambda 2s(1-s) \left[\int_1^{\frac1{1-s}} E_1(\rho\|\sigma)\d\gamma +  \int_1^\infty \gamma^{-3} E_1(\sigma\|\rho) \d\gamma \right]\frac{\d s}{s} \\
        &= \int_0^\lambda (1+s) \d s E_1(\rho\|\sigma) \\
        &= (\lambda+\frac{\lambda^2}{2}) E_1(\rho\|\sigma). 
    \end{align}  
    Putting this back in the previous equation gives the desired result. 
\end{proof}


\subsection{Yet another integral representation and $\chi^2$-bounds on divergences}

In this section, we give a final integral representation that is a quantum generalization of an expression in the proof of~\cite[Theorem 31]{george2024divergence}. 
\begin{proposition}\label{Prop:f-int-rep-chi2}
    For $f\in\cF$ we have
    \begin{align}
        D_f(\rho\|\sigma)=\int_0^1 (1-t) \left[ \frac{\partial^2}{\partial t^2} D_f\big(t\rho+(1-t)\sigma\big\|\sigma\big)\right] \d t = \int_0^1 \frac{1-t}{t^2} D_F\big(t\rho+(1-t)\sigma\big\|\sigma\big) \d t,
    \end{align}
    where $F$ is given by
    \begin{align}
        F(x) = f''(x)(x-1)^2. 
    \end{align}
\end{proposition}
\begin{proof}
    Taking integration by parts, it is readily verified that
    \begin{align}
        \int_0^1 (1-t) \left[ \frac{\partial^2}{\partial t^2} D_f\big(t\rho+(1-t)\sigma\big\|\sigma\big)\right] \d t = D_f(\rho\|\sigma) + \frac{\partial}{\partial t} D_f\big(t\rho+(1-t)\sigma\big\|\sigma\big)\Big|_{t=0} .
    \end{align}
    Recall from~\cite{hirche2023quantum} that 
    \begin{align}\label{eq:f-t-conv-comb-Gt}
        D_f(t\rho+(1-t)\sigma\|\sigma) = D_{G_t}(\rho\|\sigma), 
    \end{align}
    where $G_t(\gamma) = f(t\gamma + 1-t)$ and $G_t''(\gamma)=t^2 f''(t\gamma+1-t)$. Therefore,  
    \begin{align}
        \frac{\partial}{\partial t} D_f\big(t\rho+(1-t)\sigma\big\|\sigma\big)\Big|_{t=0} = \frac{\partial}{\partial t} D_{G_t}(\rho\|\sigma)\Big|_{t=0} = tD_{f'(t\gamma + 1-t)}(\rho\|\sigma)\Big|_{t=0} =0, 
    \end{align}
    and this gives the first equality. 
    Next, by Theorem~\ref{Thm:Df-kth-D} we have
    \begin{align}
        \frac{\partial^2}{\partial t^2} D_f\big(t\rho+(1-t)\sigma\big\|\sigma\big) = \frac{1}{t^2} D_{F(t\gamma +(1-t))}(\rho\|\sigma) = \frac{1}{t^2} D_{F} \big(t\rho+(1-t)\sigma\big\|\sigma\big) ,
    \end{align}
    where for the second equality we once again use Equation~\eqref{eq:f-t-conv-comb-Gt}. This equation immediately implies the second identity.
\end{proof}

As an application, this result can be used to derive bounds on general $f$-divergences in terms of the $\chi^2$-divergence. This, again, generalizes classical results in~\cite{george2024divergence}. 
\begin{theorem}
    Set $I=[e^{-D_{\max}(\sigma\|\rho)},e^{D_{\max}(\rho\|\sigma)}]$. Define, 
    \begin{align}
        \kappa^\uparrow_f(\rho,\sigma) &\coloneqq \max_{\substack{t\in[0,1] \\ \gamma\in I}} f''(1+t(\gamma-1)), \\
        \kappa^\downarrow_f(\rho,\sigma) &\coloneqq \min_{\substack{t\in[0,1] \\ \gamma\in I}} f''(1+t(\gamma-1)).
    \end{align}
    Then, we have
    \begin{align}
        \frac{\kappa^\downarrow_f(\rho,\sigma)}{2} \chi^2(\rho\|\sigma) \leq D_f(\rho\|\sigma) \leq \frac{\kappa^\uparrow_f(\rho,\sigma)}{2} \chi^2(\rho\|\sigma).
    \end{align}
\end{theorem}
\begin{proof}
    Our starting point is Proposition~\ref{Prop:f-int-rep-chi2}, from which we also take $F(x)$. Note that 
    \begin{align}
        D_F\big(t\rho+(1-t)\sigma\big\|\sigma\big) = D_{F_t}\big(\rho\big\|\sigma\big),
    \end{align}
    where 
    \begin{align}
        F_t(x)\coloneqq F(1-t+tx) = f''(1-t+tx) t^2 (x-1)^2. 
    \end{align}
    On the intervals $t\in[0,1]$ and $x\in I$, this can be trivially bounded by
    \begin{align}
       \kappa^\downarrow_f(\rho,\sigma) t^2 (x-1)^2  \leq F_t(x) \leq \kappa^\uparrow_f(\rho,\sigma) t^2 (x-1)^2
    \end{align}
    Hence, we have 
    \begin{align}
        \int_0^1 \frac{1-t}{t^2} D_F\big(t\rho+(1-t)\sigma\big\|\sigma\big) \d t 
        &= \int_0^1 \frac{1-t}{t^2} D_{F_t}\big(\rho\big\|\sigma\big) \d t \\
        &\leq \kappa^\uparrow_f(\rho,\sigma) \int_0^1 (1-t) \chi^2(\rho\|\sigma) \d t \\
        &= \frac{\kappa^\uparrow_f(\rho,\sigma)}{2} \chi^2(\rho\|\sigma).
    \end{align}
    This proves the upper bound in the claim and the lower bound follows similarly. 
\end{proof}

This result can, for example, be used to prove reverse Pinsker inequalities. 
\begin{corollary}\label{cor:reverse-Pinsker}
    With the definitions as above, we have 
    \begin{align}
        D_f(\rho\|\sigma) \leq \frac{2\kappa^\uparrow_f(\rho,\sigma)}{\lambda_{\min}(\sigma)} E_1(\rho\|\sigma)^2.
    \end{align}
\end{corollary}
\begin{proof}
    The main missing pieces are the inequalities 
    \begin{align}
        \chi^2(\rho\|\sigma) \leq \frac{\|\rho-\sigma\|^2_2}{\lambda_{\min}(\sigma)} \leq \frac{\|\rho-\sigma\|^2_1}{\lambda_{\min}(\sigma)}, 
    \end{align}
    where the first was proven in~\cite[Section 5.A]{audenaert2005continuity} and the second is well known. 
\end{proof}
The main difference from the reverse Pinsker inequalities previously shown for these $f$-divergences in~\cite{hirche2023quantum} is the quadratic appearance of the trace distance. Note that~\cite{george2024divergence} also generalize their classical results to the quantum setting~\cite[Section VI]{george2024divergence}, in particular to the Petz $f$-divergence by means of the so-called Nussbaum-Skola distributions~\cite{nussbaum2009chernoff}.

\subsection{Taylor expansion}

As we have seen in the previous sections, we can get all derivatives of $f$-divergences at a specific point. We use this to write down the Taylor expansion 
\begin{align}
    h(\lambda) = \sum_{i=0}^k h^{(i)}(a) \frac{(\lambda-a)^i}{i!} + R_k(\lambda),  
\end{align}
where $R_k(\lambda)$ is the remainder term for the truncated Taylor series. Since the divergences we are concerned with are well behaved, we can give the remainder term in Lagrange form as
\begin{align}
    R_k(\lambda) = h^{(k+1)}(\zeta) \frac{(\lambda-a)^{k+1}}{(k+1)!}, 
\end{align}
for some $\zeta\in[a,\lambda]$. 
Applying this expansion for the function $\lambda\mapsto D_f(\lambda\rho+(1-\lambda)\sigma\|\sigma)$ around $a=0$, we obtain 
\begin{align}
    D_f(\lambda\rho+(1-\lambda)\sigma\|\sigma) = \sum_{i=2}^k f^{(i)}(1) \, D_{(\gamma-1)^i}(\rho\|\sigma) \frac{\lambda^i}{i!} + R_k(\lambda), 
\end{align}
with the remainder term
\begin{align}\label{eq:Tayloer-remiander-term}
    R_k(\lambda) = D_{F_{k+1,\zeta}}(\rho\|\sigma) \frac{\lambda^{k+1}}{(k+1)!}, 
\end{align}
for some $\zeta\in[0,\lambda]$ where as before 
\begin{align}
    F_{k,\lambda}(\gamma)= (\gamma-1)^k f^{(k)}(\lambda\gamma+1-\lambda).
\end{align} 
Surely, numerous bounds follow from this expansion. 
As an example, we give a simple lower bound on the relative entropy that complements a well known upper bound. 
\begin{proposition}\label{prop:chi-3-D-ineq}
We have, 
\begin{align}
    \chi^2(\rho\|\sigma) - \frac13 H_3(\rho\|\sigma) \leq D(\rho\|\sigma) \leq \chi^2(\rho\|\sigma). 
\end{align}
\end{proposition}

\begin{proof}
    The upper bound is well-known, e.g.~\cite[Lemma 2.2]{gao2022complete} (see also Equation~\eqref{eq:D-chi^2-ineq} above). We prove the lower bound here. 
    Considering the Taylor expansion of the relative entropy as described above up to $k=3$, yields
    \begin{align}
        D(\lambda\rho+(1-\lambda)\sigma\|\sigma) = H_2(\rho\|\sigma)\frac{\lambda^2}{2} - (2H_3(\rho\|\sigma) -3H_2(\rho\|\sigma))\frac{\lambda^3}{6}  + R_3(\lambda).
    \end{align}
    Letting $\lambda=1$, the desired inequality follows once we show that the remainder term is non-negative. Using the definition of the remainder term in Equation~\eqref{eq:Tayloer-remiander-term}, this can be done by checking that the second derivative of $F_{4,\zeta}(\gamma)$, associated with the function $f(\gamma)= \gamma\log \gamma$, is non-negative for all $\gamma\geq 1$. This is easily verified to be true because
    \begin{align}
        F''_{4, \zeta}(\gamma) & = 6\frac{(\gamma-1)^2}{(\zeta \gamma + (1-\zeta))^3}\bigg( \frac{\zeta(\gamma-1)}{\zeta\gamma+(1-\zeta)}-1 
 \bigg)^2 .
    \end{align}
    
\end{proof}

The inequality on the left hand side in the above proposition has already been established in Equation~\eqref{eq:chi2-H3-D-log-convexity} using the log-convexity of $\alpha\mapsto \frac 1\alpha H_\alpha(\rho\|\sigma)$. Here, we give this alternative proof as it can be generalized as follows. Indeed, for any odd $k\geq3$ we have
\begin{align}
    F''_{k+1, \zeta}(\gamma) & =  \frac{k(k+1)}{(k-1)!}\frac{(\gamma-1)^{k-1}}{(\zeta \gamma + (1-\zeta))^k}\bigg( \frac{\zeta(\gamma-1)}{\zeta\gamma+(1-\zeta)}-1 
 \bigg)^2 .
\end{align}
Hence, by the above argument one could get tighter bounds on $D(\rho\| \sigma)$ if desired. For example, the case of $k=5$ yields
\begin{align}\label{eq:ineq-Taylor-k=5}
    D(\rho\|\sigma) \geq 2 \chi^2(\rho\|\sigma) - 2 H_3(\rho\|\sigma) + H_4(\rho\|\sigma) - \frac15 H_5(\rho\|\sigma). 
\end{align}
We emphasize that this latter inequality cannot be proven by the aforementioned log-convexity as it does not hold for all log-convex functions.

\paragraph*{Acknowledgments} MT is supported by the National Research Foundation, Singapore and A*STAR under its CQT Bridging Grant.

\bibliographystyle{ultimate}
\bibliography{lib}

\end{document}